\documentclass[12pt,reqno]{amsart}

\usepackage[a4paper]{geometry}
\newgeometry{margin=3cm}

\usepackage[english]{babel}
\usepackage[applemac]{inputenc}       
\usepackage[T1]{fontenc}

\usepackage{amssymb}
\usepackage{amsmath}
\usepackage{amsthm}

\usepackage{dsfont}

\usepackage{graphicx}

\usepackage{tikz}
\usetikzlibrary{calc}
\usetikzlibrary{arrows.meta}
\usetikzlibrary{positioning}
\usetikzlibrary{fit}

\theoremstyle{plain}
\newtheorem{thm}{Theorem}[section]

\newtheorem{lem}[thm]{Lemma}
\newtheorem{prop}[thm]{Proposition}

\theoremstyle{definition}
\newtheorem{defn}[thm]{Definition}
\newtheorem{rem}[thm]{Remark}

\newtheorem{exmp}[thm]{Example}

\newcommand{\ee}{{\mathrm e}}
\newcommand{\ii}{{\mathrm i}}
\newcommand{\dd}{{\mathrm d}}

\newcommand{\tr}{{\mathrm{tr}}}

\newcommand{\1}{{\mathds{1}}}

\newcommand{\pt}{\mathrm{pt}}
\newcommand{\Ch}{\mathrm{Ch}}
\newcommand{\Cr}{\mathrm{Cr}}
\newcommand{\Top}{\mathcal I}

\newcommand{\C}{\mathbb{C}}
\newcommand{\R}{\mathbb{R}}
\newcommand{\Z}{\mathbb{Z}}

\usepackage[pdftex,%
colorlinks=true,linkcolor=blue,citecolor=black,urlcolor=blue
]
{hyperref}

\begin{document}
	
\title{Eigenvalue crossings in Floquet topological systems}

\author[K. Gomi \and C. Tauber]{Kiyonori Gomi \and Cl\'ement Tauber}

\date{\today}

\begin{abstract}
The topology of electrons on a lattice subject to a periodic driving is captured by the three-dimensional winding number of the propagator that describes time-evolution within a cycle. This index captures the homotopy class of such a unitary map. In this paper, we provide an interpretation of this winding number in terms of local data associated to the the eigenvalue crossings of such a map over a three dimensional manifold, based on an idea from \cite{NathanRudner15}. We show that, up to homotopy, the crossings are a finite set of points and non degenerate. Each crossing carries a local Chern number, and the sum of these local indices coincides with the winding number. We then extend this result to fully degenerate crossings and extended submanifolds to connect with models from the physics literature. We finally classify up to homotopy the Floquet unitary maps, defined on manifolds with boundary, using the previous local indices. The results rely on a filtration of the special unitary group as well as the local data of the basic gerbe over it.
\end{abstract}

\maketitle

\section{Introduction}
 
In the context of topological insulators, it was recently realized that independent electrons on a lattice subject to a periodic drive lead to new topological phases of matter with no static counterpart \cite{Rudner13}. In these so-called Floquet topological systems, the Hamiltonian is time-dependent and periodic beyond the adiabatic approximation, so that the topology is encoded in the unitary propagator generated by Schrödinger evolution. In dimension two, the topological index for a sample without boundary ultimately relies on the computation of a three-dimensional winding number of a unitary operator over two dimensions of space and one cycle of time driving \cite{Rudner13,GrafTauber18,SadelSchulzBaldes17}.

The main assumption for such an index to be well-defined relies on the spectral properties of the propagator after one cycle of driving. If the latter has either a spectral gap, be it for perfect crystals \cite{Rudner13} or disordered systems \cite{GrafTauber18,SadelSchulzBaldes17}, or a mobility gap in the regime of strong disorder \cite{ShapiroTauber19} then $U$ can be replaced by a relative evolution that is time periodic, such that the time interval becomes the circle. Moreover, for perfect crystals with translation invariance, the Bloch decomposition reduces the problem to a map $U: T^2 \times S^1 \to U(N)$  where $T^2$ is the two dimensional Brillouin torus. The homotopy classes of such maps are characterized by the three dimensional winding number $W_3(U)$.

It is natural to ask whether the value of $W_3(U)$ can be inferred directly from the spectral properties of the map $U$, not only at finite time as above, but over the whole manifold $T^2 \times S^1$. A significant step in that direction can be found in the physics literature \cite{NathanRudner15}. It suggests that $W_3(U)$ is actually related to the eigenvalue crossings of $U$ that are topologically protected. The main result of the present paper is to provide a complete and rigorous mathematical answer to that question.

We first show that, up to homotopy, the eigenvalue crossings of such maps $U$ are always isolated points and non-degenerate. In that case, one follows a given eigenvalue of $U$ over $T^2 \times S^1$. For each crossing point there is a well-defined line bundle over a 2-dimensional submanifold surrounding it. Such line bundle carries a topological Chern number and we show that the sum of such numbers matches with $W_3(U)$. This provides a geometric interpretation of the 3-dimensional winding number in terms of 2-dimensional local indices related to the eigenvalue crossings of the unitary map.

Then we extend the scope of this theorem to operational means, in particular to deal with explicit models from the physics literature. Indeed, although the crossings are always simple up to homotopy, such a homotopy is not always explicit. Typically, the crossings occur not only at points but on any strict submanifold of $T^2\times S^1$. Moreover, the unitary propagator is at initial time always the identity matrix, a highly degenerate crossing. We show that the local index interpretation of $W_3(U)$ is also valid in that case. Finally, the unitary propagator is usually not periodic in time even if its generator is, so that the aforementioned relative construction is required for $W_3$ to be well-defined. However, for Floquet maps defined over $T^2 \times [0,1]$, the local indices are still available and we show how they actually classify such non-periodic unitary maps up to homotopy.

The proof of the main theorem relies on some explicit filtration of $SU(N)$ and on the local data of the basic gerbe over $SU(N)$ \cite{GawedzkiReis02,Meinrenken02}. Notice that a similar filtration of Hermitian matrices is used in \cite{Arnold95}. Moreover, originally developed in the context of conformal field theory, gerbes have been recently used already in the context of topological insulators, with no driving and with time-reversal symmetry \cite{Lyon15,Gawedzki15,Gawedzki17,MonacoTauber17}. The present paper is another application of this geometrical concept.

The paper is organized as follows: Section~\ref{sec:main} states the main results and illustrate them with typical examples, some of them from the physics literature. The proofs can be found in Section~\ref{sec:proofs}. Appendix~\ref{sec:reduction} provides topology arguments for homotopy classes of unitary maps, suited to Floquet systems and Appendix~\ref{app:exmp} gives further examples to illustrate the wide variety of cases that occur in the main theorems.

\medskip
\noindent\textbf{Acknowledgements.} K.G. is supported by JSPS Grant-in-Aid for Scientific Research on
Innovative Areas "Discrete Geometric Analysis for Materials Design": Grant Number JP17H06462.

\section{Local formula index for unitary maps \label{sec:main}}

Let $N\geq 2$ and consider $H:\Sigma \times S^1 \to M_N(\mathbb C)$ a family of self-adjoint matrices, where $\Sigma$ is a compact, connected and oriented two-dimensional manifold without boundary (typically $\Sigma = T^2$, the Brillouin torus). Schr\"odinger equation $\ii\partial_t U = H U$ generates a differentiable family $U: \Sigma \times [0,1]\to U(N)$ such that $U(\cdot,0)=\1$. In general $U(\cdot,1) \neq \1$ but if the latter has a spectral gap then there exists $U_\mathrm{ref}: \Sigma \times [0,1]\to U(N)$ such that $U_\mathrm{ref}(\cdot,0)=\1$ and $U_\mathrm{ref}(\cdot,1)=U(\cdot,1)$ \cite{Rudner13}. Gluing $U$ and $U_\mathrm{ref}$, the second one with the reverse orientation of $\Sigma \times [0,1]$, we end up with a map on $\Sigma \times S^1$.

Thus, up to this gluing we always work with $U: \Sigma \times S^1 \to U(N)$ differentiable, and such that $U|_{\Sigma \times \{0\}}  = \1$ for some base point $0 \in S^1$. Standard topology arguments show that homotopy classes of such maps are characterized by two topological invariants (see Appendix~\ref{sec:reduction} for details).
The first one is the one-dimensional winding number $W_1(U) \in \mathbb Z$ along $\{x\} \times S^1$ where $x \in \Sigma$ is any point.
The second one is the three-dimensional winding number $W_3(U) \in \mathbb Z$ given by
\begin{equation}\label{def_W3}
W_3(U) = \dfrac{1}{24 \pi^2} \int_{\Sigma \times S^1} \tr (U^{-1}\dd U)^3
\end{equation}
These two invariants are the obstruction classes for $U$ to be trivial, namely homotopic to the constant map while keeping the constraint $U|_{\Sigma \times \{0\}}  = \1$. Moreover, up to homotopy we can always assume $SU(N)$-valued maps in the computation of $W_3(U)$ (see also Appendix~\ref{sec:reduction}). Consequently, we focus on such maps from now on.

\subsection{Main theorem}

The main result of this paper is to provide an expression of $W_3(U)$ for $U:\Sigma \times S^1 \rightarrow SU(N)$ in terms of some local data related to the eigenvalue crossings of $U$. Below we will constantly use the following decomposition:

\begin{lem}\label{label_ev}
The eigenvalues of $V \in SU(N)$ can be written as $\ee^{2\pi \ii \lambda_1}, \ldots, \ee^{2\pi \ii \lambda_N}$ with $\lambda_i \in \mathbb R$, $\sum_i \lambda_i = 0$ and 
\begin{equation*}
\lambda_1 \geq \lambda_2 \geq \ldots \geq \lambda_N \geq \lambda_1 -1.
\end{equation*}
This writing uniquely determines the $\lambda_i$.
\end{lem}

\noindent Moreover, we shall consider two subsets of $SU(N)$
\begin{align} \label{defSUN_<1}
& SU(N)_{\leq 1} = \mathop{\bigcup}_{j=1}^N \left\lbrace V \in SU(N) \,|\, \lambda_1 > \ldots > \lambda_j \geq \lambda_{j+1} > \ldots > \lambda_N  > \lambda_1-1\right\rbrace\\
\label{defSUN_1}
& SU(N)_{1} = \mathop{\bigcup}_{j=1}^N \left\lbrace V \in SU(N) \,|\, \lambda_1 > \ldots > \lambda_j = \lambda_{j+1} > \ldots > \lambda_N > \lambda_1-1\right\rbrace
\end{align}
In the first one at most two eigenvalues coincide, whereas exactly two coincide in the second. In particular for a map $U:\Sigma \times S^1 \rightarrow SU(N)$ we get a family of eigenvalues $\lambda_i(x,t)$ for $(x,t) \in \Sigma \times S^1$, and  $ U^{-1}(SU(N)_{1})$ is the region of $\Sigma \times S^1$ where exactly two eigenvalues of $U$ cross. 

For the following let $X$ be a compact oriented 3-dimensional manifold without boundary. Typically, $X=\Sigma \times S^1$ for Floquet systems or $S^3$, $\mathbb CP^1$ for other examples below.  We claim

\begin{prop} \label{prop:ensure_general_position} Any continuous map $U : X \to SU(N)$ is homotopic to a smooth map $U' : X \to SU(N)$ such that
\begin{itemize}
\item
$U'(X) \subset SU(N)_{\leq 1}$, and
\item
 $(U')^{-1}(SU(N)_1)$ is a finite set of points.
\end{itemize}
\end{prop}

The proposition means that, up to homotopy, we can always assume at most two eigenvalues of $U$ to cross for $x \in X$. Moreover these crossings occur for a finite set of points. The proof of Proposition~\ref{prop:ensure_general_position} is given in Section~\ref{sec:filtration}. It relies on a filtration of $SU(N)$, that is constructed through the root system of its Lie algebra, and generalizes subsets \eqref{defSUN_<1} and \eqref{defSUN_1} to higher order crossings. Then an induction based on transversality arguments allows to deform $U$ to a map with the desired property. 

By Lemma~\ref{label_ev} the eigenvalues of $U$ are uniquely labelled $\lambda_j(x)$ and we define
\begin{equation*}
\Cr_j(U) = \left\lbrace x \in X \, | \,  \lambda_1 > \ldots > \lambda_j = \lambda_{j+1} > \ldots > \lambda_N > \ldots > \lambda_1-1 \right\rbrace
\end{equation*}
for $j=1,\ldots,N$. In particular $\Cr_j(U) \subset(U)^{-1}(SU(N)_1)$ is a finite set of points. Each of these crossing points carries a topological charge that we  compute as follows.

\begin{defn}\label{def:Ch_x}Let $U :X \to SU(N)$ be a smooth map such that $U(X) \subset SU(N)_{\leq 1}$ and  $(U)^{-1}(SU(N)_1)$ is a finite set of points. For $j \in \{1,\ldots,N\}$ there exists a  $3$-dimensional closed disk $D_x$ for each $x \in \Cr_j(U)$ such that: $x$ is contained in the interior of $D_x$; the eigenvalues of $U(y)$ are distinct for any $y \in D_x \backslash \{ x \}$; and $D_x \cap D_y = \emptyset$ whenever $x \neq y$. The boundary $\partial D_x$ of $D_x$ is a $2$-dimensional sphere. In particular 
	$$\mathcal{L}_x = \{ (y, v) \in \partial D_x \times \C^N |\
	U(y) v = \ee^{2\pi \ii \lambda_j(U(y))}v \}$$ forms a complex line bundle over $\partial D_x$. We define  
$\Ch(x;j) \in \mathbb Z$ as the Chern number of this line bundle. 
\end{defn}

\begin{thm}\label{thm:main}
Let  $U :X \to SU(N)$ be a smooth map such that $U(X) \subset SU(N)_{\leq 1}$ and  $(U)^{-1}(SU(N)_1)$ is a finite set of points. Then its 3-dimensional winding number reads
\begin{equation*}
W_3(U) = - \sum_{x \in \Cr_j(U)} \Ch(x;j)
\end{equation*}
for any $j=1,\ldots, N$.
\end{thm}
It should be noted that the ``$j$th crossing'' is enough for the description of $W_3(U)$, and that each one can be used equivalently.  The proof is given in Section~\ref{sec:proof_mainthm}, and relies on the existence of a collection of line bundles with connection which are part of the data of the basic gerbe over $SU(N)$. The end of the proof is a consequence of Stokes theorem.

\begin{exmp}\label{exmp:identity}Let $U : SU(2) \to SU(2)$ be the identity map. Though its $3$-dimensional winding number is clearly $W_3(U) = 1$, we apply Theorem~\ref{thm:main} to this case. Using the parametrization $u = x + \ii y$ and $v = z + \ii w$ such that $\lvert u \rvert^2 + \lvert v \rvert^2 = x^2 + y^2 + z^2 + w^2 = 1$, we specify a matrix $U \in SU(2) \cong S^3$ as
$$
U = 
\left(
\begin{array}{rr}
u & -\bar{v} \\
v & \bar{u}
\end{array}
\right).
$$
We can uniquely express the eigenvalues of $U$ as $\ee^{2\pi \ii \lambda_1}, \ee^{2\pi \ii \lambda_2}$ in terms of $\lambda_1 = \theta$ and $\lambda_2 = - \theta$, where $\theta \in [0, 1/2]$ is subject to $x = \cos 2\pi \theta$. If $\theta \neq 0, 1/2$, then the eigenvalues are distinct. We clearly have $\Cr_1(U) = \1$ and $\Cr_2(U) = -\1$ where $\1 \in SU(2)$ is the identity matrix.
Thus the identity map $U : SU(2) \to SU(2)$ satisfies the assumptions in Theorem~\ref{thm:main}. 

A possible disc $D_\1$ that surrounds $\1 \in \Cr_1(U)$ is the one where $x\geq 0$. Its boundary $\partial D_\1$ corresponds to $x=0$, namely $u=\ii y$ and $v = z + \ii w$ with $y^2 + z^2 + w^2 = 1$. There, $U$
has the distinct eigenvalues $\ii = \ee^{\pi \ii/4}$ and $-\ii = \ee^{-\pi \ii /4}$. Because $\lambda_1 =\tfrac{1}{4} > \lambda_2 =-\tfrac{1}{4} >\lambda_1-1 =-\tfrac{3}{4}$,
the local index $\Ch(x=\1; 1)$ is the Chern number of the line bundle over $\partial D_\1$ whose fibers are the eigenspaces with eigenvalues $\ee^{2\pi \ii \lambda_1} = \ii$. If $y \neq 1$ (resp. $y\neq -1$), then the eigenvector $v_-(U)$ (resp.  $v_+(U)$) of $U \in \partial D_\1$ with eigenvalue $\ii$ is
$$
v_-(U) =
\left(
\begin{array}{c}
\frac{\ii \bar{v}}{1 - y} \\ 1
\end{array}
\right), \qquad 
v_+(U) =
\left(
\begin{array}{c}
1 \\
\frac{-\ii v}{1 + y}
\end{array}
\right),
$$
where $v = z + \ii w$. 
These eigenvectors give local frames of the line bundle $L \to \partial D_\1$. On the circle in $\partial D_\1$ parametrized by $y = 0$ and $v = z + \ii w \in S^1$
we have a $U(1)$-valued map $f(U) = -\ii v$ which measures the discrepancy of $v_+(U)$ and $v_-(U)$ by $v_+(U) = f(U)v_-(U)$. The first Chern number of $L \to \partial D_\1$ agrees with the winding number of $f$, which is $1$, provided that the orientation on $\partial D_\1$ is induced from the sphere $\{ (y, z, w) \in \mathbb R^3 |\ y^2 + z^2 + w^2 = 1 \}$. However, the orientation on $\partial D_\1$ induced from $SU(2)$ is opposite. Therefore we get $\Ch(x;1) = -1$ at $x = \1 \in \Cr_1(U)$, which reproduces $W_3(U) = - \Ch(x;1) = 1$. 

A convenient choice of disk $D_{-\1}$ around $-\1 \in \Cr_2(U)$ is to take $x \leq 0$. Its boundary $\partial D_{-\1}$ coincides with $\partial D_{\1}$ but with reverse orientation. Focusing on $\lambda_2$, we look instead at the eigenvector of $U$ associated to the eigenvalue $-\ii$. All together this leads similarly to $\Ch(x',2) = \Ch(x,1)= -W_3(U)$ for $x'=-\1$ and $x=\1$, as expected.

\end{exmp}

\subsection{An operational version}

Together with Proposition~\ref{prop:ensure_general_position}, Theorem~\ref{thm:main} gives a general relation between $W_3(U)$ and the topological charges of eigenvalue crossings of $U$. However for operational means, it is relevant to  extend its scope. Indeed it is not always easy to continuously deform a given $U$ to a map as in Proposition~\ref{prop:ensure_general_position}. In particular when dealing with models from the physics literature, see below. Two features naturally occur. 

First, eigenvalues usually cross on extended submanifolds rather than points.  The topological charge of such crossing can be defined similarly to Definition~\ref{def:Ch_x}.

\begin{lem}\label{def:Ch_X}
Let  $U :X \to SU(N)$ smooth such that $U(X) \subset SU(N)_{\leq 1}$ and for some $j$  the subspace $\Cr_j(U) = \bigsqcup_a X_a$ is the disjoint union of compact, connected and orientable submanifolds $X_a \subset X$ of dimension $\dim X_a < 3$ without boundary. Then we can choose a (closed) tubular neighborhood $N_a$ of each $X_a$ such that: the eigenvalues of $U(y)$ are distinct for any $y \in N_a \backslash X_a$; and $N_a \cap N_{a'} = \emptyset$ whenever $a \neq a'$. Then $\partial N_a$ is an oriented 2-dimensional manifold. Let $\mathcal{L}_a \to \partial N_a$ be the line bundle over $\partial N_a$  whose fiber is the eigenspace of $U$ with eigenvalue $\ee^{2\pi \ii \lambda_j(U)}$. We define by $\Ch(X_a;j) \in \mathbb Z$  the Chern number of this line bundle. 
\end{lem}
\begin{proof}
We can identify $N_a$ as a disk bundle over $X_a$, whose rank is $r_a = 3 - \dim X_a$, in which $X_a \subset N_a$ is identified with the image of the zero section. Since $X$ and $X_a$ are orientable, so is $N_a \to X_a$ as a normal bundle. Thus, in view of the classification of orientable real vector bundle in low dimensions, we find that the normal bundle in question is always trivial: $N_a \cong X_a \times D^{r_a}$, where $D^{r_a}$ is the closed $r_a$-dimensional disk. As a consequence, we get $\partial N_a \cong X_a \times \partial D^{r_a}$, which is a  $2$-dimensional manifold. The submanifold $N_a \subset X$ inherits an orientation from $X$. We choose the orientation on $\partial N_a$ to be that induced from $N_a$. The submanifold $N_a \subset X$ inherits an orientation from $X$. We choose the orientation on $\partial N_a$ to be that induced from $N_a$.
\end{proof}

\begin{rem}\label{rk:submanifolds}
Concretely, if $\dim X_a = 0$ then $X_a = \pt$ and $\partial N_a \cong S^2$. If $\dim X_a = 1$ then $X_a \cong S^1$ and $\partial N_a \cong S^1 \times S^1$. And if $\dim X_a = 2$  then $\partial N_a \cong X_a \sqcup X_a$. In the latter case $N_a \cong X_a \times [-1, 1]$. We can give an orientation to $X_a$ so that the orientation on $X_a \times [-1, 1]$ agrees with that on $N_a$ induced from $X$. If we write $X_a^\pm = X_a \times \{ \pm 1 \}$, then $\partial N_a \cong X_a^+ \sqcup X_a^-$, and the orientation on $X_a^-$ is the same as that on $X_a$, whereas that on $X_a^+$ is opposite. In these notations $\Ch(X_a,j)$ is expressed as a difference of two Chern numbers between line bundles over $X_a^-$ and $X_a^+$.
\end{rem}

%
%

The second generalization is motivated by Floquet systems with $X=\Sigma \times S^1$ and where one always has $U(x,0) = \1$ for any $x\in \Sigma$ and some base point $0\in S^1$. This is a highly-degenerate crossing occurring on a 2-dimensional submanifold $\Sigma \times \{0\} \cong \mathbb T^2$. We can acually assume that $U^{-1}(\1) = \bigsqcup_b Y_b$ is the disjoint union of compact, connected and orientable submanifolds $Y_b \subset X$ of dimension $\dim Y_b < 3$ without boundary. A similar tubular neighborhood $N_b$ can be defined with a line bundle over $\partial N_b$ and a corresponding Chern number $\Ch(Y_b,j)$,  that is well defined for each $j$. 
  
 \begin{thm} \label{thm:full_degeneracy_at_1}
Let  $U : X \to SU(N)$ be a smooth map such that
\begin{itemize}
\item
$U(X) \subset SU(N)_{\le 1} \cup \{ \1 \}$;

\item
there is $j \in \{ 1, \cdots, N \}$ such that $\Cr_j(U) = \bigsqcup_a X_a$ is the disjoint union of a finite number of compact, connected and orientable submanifolds $X_a \subset X$ of $\dim X_a < 3$ without boundary.

\item
the subspace $U^{-1}(\1) = \bigsqcup_b Y_b$ is the disjoint union of compact, connected and orientable submanifolds $Y_b \subset X$ of $\dim Y_b < 3$ without boundary.
\end{itemize}
Then 
$$
W_3(U) = - \sum_{a} \Ch(X_a; j) + \sum_b \sum_{\ell =j+1}^N \Ch(Y_b; \ell),
$$
with the latter sum vanishing when $j=N$ or $N=2$ by convention. 
\end{thm}

Notice that, for the identity matrix one has $\lambda_1 = \ldots = \lambda_N = 0$ but $\lambda_N > \lambda_1 -1 =-1$ so that the crossing where $U=\1$ is indeed ignored when computing $W_3(U)$ through the $N$-th eigenvalue. Moreover if $N=2$ then $U^{-1}(\1) \subset \Cr_1(U)$, namely this crossing is non degenerate in that case and the second term is absent. 

Several examples with various contributions from $\Cr_j(U)$ and $U^{-1}(\1)$ can be found in Appendix~\ref{app:exmp}. Here we provide one that originally comes from the physical model of \cite{Rudner13}. 

\begin{exmp}\label{exmp:rudner}
The Hamiltonian $H : \mathbb T^2 \times S^1 \rightarrow M_2(\mathbb C)$ is a two-band model that is piecewise constant in time: for $i = 1,\ldots 4$, 
$
H(t,k) = H_i(k)$ for $\tfrac{i-1}{4} \leq t < \tfrac{i}{4}$. For $k = (k_1,k_2) \in [-\pi,\pi]^2$ (we identify $\mathbb T^2$ with its fundamental domain)
\begin{align}\label{Rudner_hamiltonian}
& H_1(k_1,k_2) = - J \sigma_1 \cr
& H_2(k_1,k_2) = - J \big( \cos(k_1-k_2) \sigma_1 + \sin(k_1-k_2) \sigma_2 \big) \cr
& H_3(k_1,k_2) = -J \big( \cos(k_1) \sigma_1 + \sin(k_1) \sigma_2 \big) \cr
& H_4(k_1,k_2) = -J  \big( \cos(k_2) \sigma_1 + \sin(k_2) \sigma_2 \big)
\end{align}
where $\sigma_1$, $\sigma_2$ and $\sigma_3$ are the Pauli matrices:
$$
\sigma_1 = \begin{pmatrix} 0 & 1 \\ 1 & 0
\end{pmatrix}, \qquad \sigma_2 = \begin{pmatrix} 0 & -\ii \\ \ii & 0
\end{pmatrix},\qquad
\sigma_3 = \begin{pmatrix} 1 & 0 \\ 0 & -1
\end{pmatrix}.
$$
The unitary propagator is computed via 
the exponentials $U_i(k_1,k_2,t):=\ee^{-\ii t H_i(k_1,k_2)}$ and reads $U(k_1,k_2,t) = U_i(k_1,k_2,t-\tfrac{i-1}{4}) U_{i-1}(k_1,k_2,\tfrac{1}{4}) \ldots U_1 (k_1,k_2,\tfrac{1}{4})$ for $ \tfrac{i-1}{4} \leq t < \tfrac{i}{4}$. It has a simple explicit expression if we consider the resonant case where $J=2\pi$
\begin{equation}\label{U_explicit}
U(k_1,k_2,t) = \left\lbrace 
\begin{array}{ll}
\cos(2\pi t) \1 + \ii \sin(2\pi t) \sigma_1,& 0\leq t \leq \frac{1}{4},\\
\cos(2\pi t)\big(\cos(k_1-k_2) \1 - \ii \sin(k_1-k_2) \sigma_3 \big) + \ii \sin(2\pi t) \sigma_1,& \frac{1}{4} \leq t \leq \frac{1}{2},\\
\cos(2\pi t)\big(\cos(k_1-k_2) \1 - \ii \sin(k_1-k_2) \sigma_3 \big) &\\
 \hspace {3.8cm}+ \ii \sin(2\pi t) (\cos(k_2) \sigma_1 + \sin(k_2) \sigma_2) ,& \frac{1}{2} \leq t \leq \frac{3}{4},\\
 \cos(2\pi t) \1 + \ii \sin(2\pi t) (\cos(k_2) \sigma_1 + \sin(k_2) \sigma_2) ,& \frac{3}{4} \leq t \leq 1.
\end{array}
\right.
\end{equation}
In particular $U(\cdot,0)=\1$, then for $i=1,\ldots,4$, $\det(U_i(k,t))=1$ so that $U(k,t) \in SU(2)$ and finally $U(\cdot,1) =\1$ so that $U$ is well-defined on $\mathbb T^2\times S^1$. A direct computation of \eqref{def_W3} shows that $W_3(U)=1$. Following the decomposition of Lemma~\ref{label_ev}, the eigenvalues of $U$ are given in Table~\ref{tab:ev_exmp} and schematically represented in Figure~\ref{fig:exmp}(a).
\begin{table}[htb]
\centering
\begin{tabular}{|c|c|c|c|}
	\hline
& $0\leq t \leq 1/4$ & $1/4 \leq t \leq 3/4$ & $3/4 \leq t \leq 1$ \\ \hline
$\lambda_1$ & $t$ & $(2\pi)^{-1}\arccos\big( \cos(k_1-k_2) \cos(2\pi t) \big)$ & $1-t$ \\ \hline
$\lambda_2$ & $-t$ & $-(2\pi)^{-1}\arccos \big( \cos(k_1-k_2) \cos(2\pi t) \big)$ & $t-1$  \\ \hline
\end{tabular}
\caption{Eigenvalues $\ee^{2 \pi \ii \lambda_1}$ and $\ee^{2 \pi \ii \lambda_2}$ of the map U given by \eqref{U_explicit}.\label{tab:ev_exmp}}
\end{table}

The crossing $\lambda_1 = \lambda_2$ occurs at $t=0$ (and $t=1$) for arbitrary $k_1,k_2$  and at $t=1/2$ for $k_1-k_2=0\mod 2\pi$. In both cases $\lambda=0$, corresponding to $\ee^{2\pi \ii \lambda}=1$.  The associated subspace crossing is 
$
\mathrm{Cr}_1(U) =  \mathbb T^2 \times \{0\} \sqcup S_1 \times \{1/2\}
$
where $S_1 = \{(\ee^{\ii k_1},\ee^{\ii k_1}) \} \subset \mathbb T^2$. As mentioned before, since $N=2$ the degenerate crossing due to $U(\cdot,0) = U(\cdot,1)=\1$ is just a single eigenvalue crossing and hence part of $\mathrm{Cr}_1(U)$. The crossing $\lambda_2 = \lambda_1 -1$ occurs at $t=1/2$ for $k_1-k_2=\pi\mod 2 \pi$, where $\lambda_2 = -1/2$ corresponding to $\ee^{2\pi \ii \lambda}=-1$.  The associated subspace crossing is
$
\mathrm{Cr}_2(U) =  S'_1 \times \{1/2\}
$
where $S'_1 = \{(\ee^{\ii k_1},\ee^{-\ii k_1}) \} \subset \mathbb T^2$. 
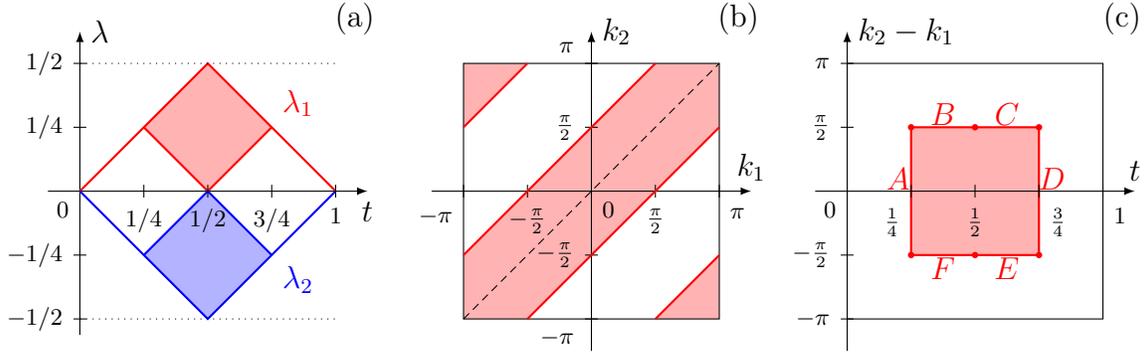
\begin{figure}[htb]
\centering
\begin{tikzpicture}[scale=0.85]
\draw[-latex] (-0.5,0) -- (4.5,0) node[below]{$t$};
\draw[-latex] (0,-2.25) -- (0,2.5) node[right]{$\lambda$};
\fill[red!30] (1,1) -- (2,0) -- (3,1) -- (2,2)--cycle;
\draw[thick,red] (0,0) -- (2,2) -- (4,0);
\draw[thick,red] (1,1) -- (2,0) -- (3,1);
\fill[blue!30] (1,-1) -- (2,0) -- (3,-1) -- (2,-2)--cycle;
\draw[thick,blue] (0,0) -- (2,-2) -- (4,0);
\draw[thick,blue] (1,-1) -- (2,0) -- (3,-1);
\draw(0,0) node[below left] {\scriptsize $0$};
\draw (1,0.1)--(1,-0.1) node[below]{\scriptsize $1/4$};
\draw (3,0.1)--(3,-0.1) node[below]{\scriptsize $3/4$};
\draw (2,0.1)--(2,-0.1) node[below]{\scriptsize $1/2$};
\draw (4,0.1)--(4,-0.1) node[below]{\scriptsize $1$};
\draw (0.1,1)--(-0.1,1) node[left]{\scriptsize $1/4$};
\draw (0.1,2)--(-0.1,2) node[left]{\scriptsize $1/2$};
\draw (0.1,-1)--(-0.1,-1) node[left]{\scriptsize $-1/4$};
\draw (0.1,-2)--(-0.1,-2) node[left]{\scriptsize $-1/2$};
\draw[dotted] (0,2) -- (4,2);
\draw[dotted] (0,-2) -- (4,-2);
\draw[red] (3,1) node[above right]{$\lambda_1$};
\draw[blue] (3,-1) node[below right]{$\lambda_2$};

\draw (4.3,2.7) node {(a)};

\begin{scope}[xshift=8cm]
\fill[red!30] (-2,-2) -- (-2,-1) -- (1,2) -- (2,2) -- (2,1) -- (-1,-2) -- cycle;
\fill[red!30] (-2,2) -- (-1,2) -- (-2,1) -- cycle;
\fill[red!30] (2,-2) -- (1,-2) -- (2,-1) -- cycle;
\draw[-latex] (-2.5,0) -- (2.5,0) node[above]{$k_1$};
\draw[-latex] (0,-2.5) -- (0,2.5) node[right]{$k_2$};
\draw (-2,2) -- (2,2) -- (2,-2) -- (-2,-2)--cycle;
\draw(0,0) node[below right] {\scriptsize $0$};
\draw (1,0.1)--(1,-0.1) node[below]{\scriptsize $\frac{\pi}{2}$};
\draw (-1,0.1)--(-1,-0.1) node[below]{\scriptsize $-\frac{\pi}{2}$};
\draw (2,0.1)--(2,-0.1) node[below right]{\scriptsize $\pi$};
\draw (-2,0.1)--(-2,-0.1) node[below left]{\scriptsize $-\pi$};
\draw (0.1,1)--(-0.1,1) node[left]{\scriptsize $\frac{\pi}{2}$};
\draw (0.1,2)--(-0.1,2) node[above left]{\scriptsize $\pi$};
\draw (0.1,-1)--(-0.1,-1) node[left]{\scriptsize $-\frac{\pi}{2}$};
\draw (0.1,-2)--(-0.1,-2) node[below left]{\scriptsize $-\pi$};
\draw[densely dashed] (-2,-2) -- (2,2);
\draw[thick,red] (-2,-1) -- (1,2);
\draw[thick,red] (-1,-2) -- (2,1);
\draw[thick,red] (-2,1) -- (-1,2);
\draw[thick,red] (1,-2) -- (2,-1);

\draw (2.3,2.7) node {(b)};
\end{scope}

\begin{scope}[xshift=12cm]
\draw[red,thick,fill=red!30] (1,1) -- (3,1) -- (3,-1) -- (1,-1) -- cycle;
\draw[-latex] (-0.5,0) -- (4.5,0) node[above]{$t$};
\draw[-latex] (0,-2.5) -- (0,2.5) node[right]{$k_2-k_1$};
\draw (0,2) -- (4,2) -- (4,-2) -- (0,-2);
\draw(0,0) node[below left] {\scriptsize $0$};
\draw (1,0.1)--(1,-0.1) node[below left]{\scriptsize $\frac{1}{4}$};
\draw (3,0.1)--(3,-0.1) node[below right]{\scriptsize $\frac{3}{4}$};
\draw (2,0.1)--(2,-0.1) node[below]{\scriptsize $\frac{1}{2}$};
\draw (4,0.1)--(4,-0.1) node[below right]{\scriptsize $1$};
\draw (0.1,1)--(-0.1,1) node[left]{\scriptsize $\frac{\pi}{2}$};
\draw (0.1,2)--(-0.1,2) node[left]{\scriptsize $\pi$};
\draw (0.1,-1)--(-0.1,-1) node[left]{\scriptsize $-\frac{\pi}{2}$};
\draw (0.1,-2)--(-0.1,-2) node[left]{\scriptsize $-\pi$};
\draw[red] (0.8,0.2) node{$A$};
\draw[red] (1.5,1.2) node{$B$};
\draw[red] (2.5,1.2) node{$C$};
\draw[red] (3.2,0.2) node{$D$};
\draw[red] (2.5,-1.2) node{$E$};
\draw[red] (1.5,-1.2) node{$F$};
\fill[red] (1,1) circle (0.05cm);
\fill[red] (2,1) circle (0.05cm);
\fill[red] (3,1) circle (0.05cm);
\fill[red] (1,-1) circle (0.05cm);
\fill[red] (2,-1) circle (0.05cm);
\fill[red] (3,-1) circle (0.05cm);

\draw (4.3,2.7) node {(c)};
\end{scope}
\end{tikzpicture}
\caption{(a) Eigenvalues of Table~\ref{tab:ev_exmp}. (b,c) Tubular neighborhood for $\lambda_1 = \lambda_2$ at $t=\tfrac{1}{2}$. \label{fig:exmp}}
\end{figure}

We first focus on $\Cr_1(U)$. For $X_0=\mathbb T^2 \times \{0\}$ we take a tubular neighborhood to be $N_0 = \mathbb T^2 \times [0,1/8] \cup [7/8,0]$. Its boundary is $\partial N_0 = \mathbb T^2 \times \{1/8\} -\mathbb T^2 \times \{7/8\}$. On each boundary $\lambda_1 = -\lambda_2 = 1/8$. Moreover, from \eqref{U_explicit} we notice that $U$ is independent of $k_2$ (resp. $k_1$) at $t=1/8$ (resp. $t=7/8$), and so are the corresponding eigenvectors. Hence the corresponding line bundle associated to $\lambda_1$ over $\mathbb T^2$ is trivial, so that $\Ch(X_0,1) = 0$. For $X_a = S_1 \times \{1/2\}$ we take the following tubular neighborhood
\begin{equation*}
N_a = \left\{ (\ee^{\ii k_1}, \ee^{\ii k_2}) \, | \, k_1 \in [-\pi,\pi], k_2 \in [k_1-\tfrac{\pi}{2},k_1+\tfrac{\pi}{2}]\right\} \times [-\tfrac{1}{4},\tfrac{1}{4}];
\end{equation*}
Its fundamental domain is represented in Figure~\ref{fig:exmp}(b) and (c). Its boundary is the gluing of four pieces
\begin{align*}
\partial N_a = &\, \left\{ (\ee^{\ii k_1}, \ee^{\ii (k_1 \pm \frac{\pi}{2})}\,|\, k_1 \in [-\pi,\pi] \right\} \times  [-1/4,1/4] \cr
& \cup \left\{ (\ee^{\ii k_1}, \ee^{\ii k_1 }\,|\, k_1 \in [-\pi,\pi] , k_2 \in [k_1-\tfrac{\pi}{2},k_1+\tfrac{\pi}{2}]\right\} \times \{\pm \tfrac{1}{4}\}
\end{align*}
so that $\partial N_a \cong \mathbb T^2$. Moreover on $\partial N_a$ the eigenvalues are constant $\lambda_1 = - \lambda_2 = \tfrac{1}{4}$. Thus $\Ch(X_a,1)$ is computed through Chern number of the line bundle associated to eigenvalue $\ee^{2\pi \ii \lambda_1}=\ii$ over $\partial N_a$. We compute it by local sections and the obstruction to glue them all together. Each section corresponds to one of the pieces described on the figure above ($A, \ldots, F$). For convenience we also introduce $\alpha: [\tfrac{1}{4},\tfrac{3}{4}] \mapsto [-1,1]$ where
\begin{equation*}
\alpha(t) = \dfrac{1 + \cos(2\pi t)}{\sin(2\pi t)} = \dfrac{\sin(2\pi t)}{1 - \cos(2\pi t)}.
\end{equation*}
In particular $\alpha(\tfrac{1}{4}) = 1$, $\alpha(\tfrac{1}{2}) = 0$ and $\alpha(\tfrac{3}{4}) = -1$. The eigenvectors for $\ii$ of $U$ are obtained from \eqref{U_explicit} and read
\begin{align*}
\psi_A = \tfrac{1}{\sqrt{2}} \begin{pmatrix}
1 \\ 1
\end{pmatrix}, \qquad & \psi_B = \tfrac{1}{\sqrt{1+\alpha^2(t)}} \begin{pmatrix}
 \alpha(t)\\ 1
\end{pmatrix},
& \psi_C = \tfrac{1}{\sqrt{1+\alpha^2(t)}} \begin{pmatrix}
-\ii \ee^{-\ii k_1} \alpha(t)\\ 1
\end{pmatrix} \cr
\psi_D = \tfrac{1}{\sqrt{2}} \begin{pmatrix}
-\ee^{-\ii k_2} \\ 1 \end{pmatrix}, \qquad & \psi_E = \tfrac{1}{\sqrt{1+\alpha^2(t)}} \begin{pmatrix}
\ii \ee^{-\ii k_1} \\ \alpha(t)
\end{pmatrix},& \psi_F = \tfrac{1}{\sqrt{1+\alpha^2(t)}} \begin{pmatrix}
1 \\ \alpha(t)
\end{pmatrix}.
\end{align*}
One easily checks that all sections coincide at their transition, except for $\psi_F=\ii \ee^{-\ii k_1}\psi_E$ at $t=\tfrac{1}{2}, k_2 = k_1 - \tfrac{\pi}{2}$.
Thus the gluings are all trivial except between $F$ and $E$ that has a $k_1$-dependent discrepancy $f(k_1) = \ii \ee^{-\ii k_1}$. The Chern number over $\partial N_a$ is given by the winding number of this discrepancy, so that $\Ch(X_a,1) = W_1(f) = -1$. Thus $W_3[U] = - \Ch(X_0,1) - \Ch(X_a,1)$ as expected.

The computation with $\Cr_2(U) =S'_1 \times \{1/2\}=X_b$ is analogous. We take the following tubular neighborhood
\begin{equation*}
N_b = \left\{ (\ee^{\ii k_1}, \ee^{\ii k_2}) \, | \, k_1 \in [-\pi,\pi], k_2 \in [k_1-\pi,k_1-\tfrac{\pi}{2}] \cup [k_1+\tfrac{\pi}{2},k_1+\pi]\right\} \times [-\tfrac{1}{4},\tfrac{1}{4}].
\end{equation*}
Its fundamental domain is the complement of the one represented in Figure~\ref{fig:exmp}(b). On its boundary $\partial N_b$ the eigenvalues are constant $\lambda_1 = - \lambda_2 = \tfrac{1}{4}$, corresponding to $\ee^{2\pi \ii \lambda_2}=-\ii$. Thus $\Ch(N_b,2)$ is computed through the Chern number of corresponding line bundle. A computation similar to the previous one leads to $\Ch(X_b,2) = -1$ as expected.
\end{exmp}

\subsection{Non-periodic case and classification of time-driven systems}

The last issue is to classify, up to homotopy, unitary maps that are defined on manifold with boundary. More precisely, the unitary propagator is usually not periodic in time, even if its generator is. In that case the winding number $W_3$ is not well defined, but the eigenvalue crossing Chern numbers are and actually determine the homotopy classes of such maps. Similarly to \eqref{defSUN_<1} we define 
\begin{equation}
 SU(N)_{\leq 0} =  \left\lbrace V \in SU(N) \,|\, \lambda_1  > \lambda_2 > \ldots > \lambda_N > \lambda_1-1\right\rbrace
\label{defSUN_0}
\end{equation}
where there is no eigenvalue crossing.

\begin{defn}
A Floquet map is a smooth map $U : \Sigma \times [0,1] \to SU(N)$ such that:
\begin{enumerate}
\item $U|_{\Sigma \times \{0\}} = \1$ and $\exists \epsilon >0,  U(\Sigma \times (0,\epsilon)) \in SU(N)_{\leq 0}$
\item $U(\Sigma \times  \{1\}) \subset SU(N)_{\leq 0}$
\item $U(\Sigma \times (0,1)) \subset SU(N)_{\leq 1}$ and for each $j=1,\ldots, N$, $\Cr_j(U) = \bigsqcup_a X_a$ is the disjoint union of compact, connected and orientable submanifolds $X_a \subset X$ of dimension $\dim X_a < 3$ without boundary.
\end{enumerate}
\end{defn}
Notice that the first two conditions naturally occur in physical systems\footnote{Condition 2 requires that the eigenvalues of $U(\Sigma \times \{1\})$ are well separated by ``local'' spectral gaps, which is slightly more general than the global spectral gap property that is usually assumed.} and the third one can always be assumed up to homotopy, similarly to Proposition~\ref{prop:ensure_general_position}. See Remark~\ref{rem:deformation_boundary} below. 
\begin{defn}\label{def:top_floquet}
Let $U : \Sigma \times [0,1] \to SU(N)$ be a Floquet map. The topological indices are, for $j=1,\ldots, N$
\begin{equation*}
\mathcal I(U;j) = - \sum_a \Ch(X_a;j) + \sum_{\ell = j+1}^N \Ch(\Sigma \times \{0\};\ell),
\end{equation*}
where $\Ch(\Sigma \times \{0\},\ell)$ is the Chern number of the $\ell$-th eigenvector line bundle of $U$ over a collar neighborhood of $\Sigma \times \{0\} \subset \Sigma \times [0,1]$. The second sum above vanishes when $j=N$ or $N=2$.
\end{defn}
In contrast to periodic map where Theorem~\ref{thm:full_degeneracy_at_1} provides an equality between all the $\mathcal I(U;j)$, this is not the case for Floquet maps. The indices are however not completely independent. 
\begin{prop}\label{prop:rel_I_C}
Let $U : \Sigma \times [0,1] \to SU(N)$ be a Floquet map and $V = U|_{\Sigma \times \{1\}}$. For $p<q$,
\begin{equation*}
\mathcal I(U;q) - \mathcal I (U;p) = \sum_{\ell = p+1}^q C(V;\ell),
\end{equation*}
where $C(V;\ell)$ is the Chern number of the $\ell$-th eigenvector line bundle of $V$ over $\Sigma \times \{1\}$.
\end{prop}
Since $V : \Sigma \times \{1\} \to SU(N)_{\leq 0}$ the line bundles are well defined directly over $\Sigma \times \{1\}$ rather than some neighborhood of it, hence we use a different notation for the Chern numbers not to confuse them with $\Ch$. We are finally able to classify Floquet maps up to homotopy.

\begin{thm}\label{thm:Floquet_maps} Let $U_0, U_1 : \Sigma \times [0,1] \to SU(N)$ be two Floquet maps. The following statements are equivalent:
\begin{enumerate}
\item $U_0$ and $U_1$ are relative homotopic: it exists $\tilde U : \Sigma \times [0,1] \times [0,1] \to SU(N)$ smooth such that $\tilde U|_{\Sigma \times [0,1] \times \{i\}} = U_i$ for $i=0,1$; $\tilde U|_{\Sigma \times \{0\} \times [0,1] } =\1$ and $\tilde U(\Sigma \times \{1\} \times [0,1]) \subset SU_{N \leq 0}$.
\item $\mathcal I(U_0;j) = \mathcal I(U_1;j)$ for all $j =1,\ldots,N$.
\item $\mathcal I(U_0;j) = \mathcal I(U_1;j)$ for some  $j \in 1,\ldots,N$ and $C(V_0;p) = C(V_1;p)$ for all $p =1,\ldots,N$.
\end{enumerate}
\end{thm}

Thus, in contrast to periodic maps, homotopy classes of Floquet maps are characterized by a set of $N$ indices. Finally note that if two Floquet maps have the same endpoint the previous theorem simplifies to
\begin{prop}\label{prop:Floquet_maps}
	Let $U_0, U_1 : \Sigma \times [0,1] \to SU(N)$ be two Floquet maps such that $U_0|_{\Sigma \times \{1\}} = U_1|_{\Sigma \times \{1\}}$. Then $U_0$ and $U_1$ are homotopic relative to $\Sigma \times \partial [0,1]$ iff $\mathcal I(U_0;j) = \mathcal I(U_1;j)$ for some (and hence all) $j \in \{1,\ldots,N\}$. Moreover $\mathcal I(U_0;j) - \mathcal I(U_1;j) = W_3(U_0 \cup U_1)$ where $U_0 \cup U_1 : \Sigma \times S^1$ is the map given by gluing $U_0$ and $U_1$.
\end{prop}

We actually first prove this proposition in Section~\ref{sec:proof_floquet} then use it to prove Theorem~\ref{thm:Floquet_maps}. An example with a non-trivial contribution from $\Sigma \times \{0\}$ is provided in Appendix~\ref{subsec:example_top}. Most of the models from the physics literature provides Floquet maps, but only a few of them are simple enough to pursue analytical computations of the topological indices. Here we recycle the physical model from Example~\ref{exmp:rudner} to avoid too heavy computations.

\begin{exmp}
Consider the piecewise constant Hamiltonian \eqref{Rudner_hamiltonian} with $J=2\pi$ but only for $t \in [0, \tfrac{5}{8}]$. The evolution is given by \eqref{U_explicit} with $t \in [0, \tfrac{5}{8}]$. This restriction is a Floquet map. Indeed one has
\begin{equation*}
V(k_1,k_2)=U(k_1,k_2,\tfrac{5}{8})  = -\tfrac{1}{\sqrt 2} \begin{pmatrix}\ee^{\ii(k_1-k_2)} & \ii \ee^{\ii k_2} \\ \ii \ee^{-\ii k_2} & \ee^{-\ii(k_1-k_2)}\end{pmatrix} \in SU(2)_{\leq 0},
\end{equation*}
which can be seen from Table~\ref{tab:ev_exmp} and Figure~\ref{fig:exmp}(a). The local charges have been already computed in Example~\ref{exmp:rudner}. One has $\Cr_1(U)= X_0 \sqcup X_a$ where $X_0 =T^2 \times \{0\}$ and $X_a = S^1 \times \{\tfrac{1}{2}\}$, with $\Ch(X_0;1) =0$ and $\Ch(X_a;1)=-1$, and $\Cr_2(U)=X_b = S'_1 \times  \{\tfrac{1}{2}\}$ with $\Ch(X_b;2)=-1$. Thus $\Top(U;1) = \Top(U;2) =1$ even though $W_3(U) \notin \mathbb Z$.

The fact that the two indices coincide comes from the vanishing of the Chern number of $V$, according to Proposition~\ref{prop:rel_I_C}. Indeed, the eigenvector associated to eigenvalue $\lambda_1$ (given in Table~\ref{tab:ev_exmp}) of $U(k_1,k_2,\tfrac{5}{8})$ is
\begin{equation*}
\begin{pmatrix}
\ee^{\ii k_2}\big(\sin(k_1-k_2) - \sqrt{2-\cos^2(k_x-k_y)}\big) \\ 1
\end{pmatrix}.
\end{equation*}
This is a non-vanishing regular section over $T^2$ so that $C(V;1)=0$. Similarly $C(V;2)=0$.

Similarly, one could consider the restriction of the same model but for $t \in [0, \tfrac{3}{8}]$ instead. This gives $\Top(U;1) = \Top(U;2) =0$ and $C(V;1)=C(V;2)=0$. Although these two examples appear homotopic one to each other through \eqref{U_explicit}, this is actually not the case because the $SU(N)_{\leq 0}$-valued property of the end point is not preserved by such map: it is explicitly broken at $t=\tfrac{1}{2}$.
\end{exmp}

\section{Proofs \label{sec:proofs}}

\subsection{Filtration of SU(N)}

In this section we prove Proposition~\ref{prop:ensure_general_position}.

\subsubsection{Review of root system}

To analyze the eigenvalues of $SU(N)$ and also to use results in \cite{GawedzkiReis02,Meinrenken02} later, we review some notations related to the standard root system of $SU(N)$. Let $T \subset SU(N)$ be the subgroup consisting of diagonal matrices, which gives rise to a maximal torus. We write $\mathfrak{su}(N)$ and $\mathfrak{t}$ for the Lie algebra of $SU(N)$ and $T$, respectively, and $\mathfrak{t}^* = \mathrm{Hom}(\mathfrak{t}, \R)$ its dual. We let 
$$
\langle \ , \ \rangle : \ \mathfrak{su}(N) \times \mathfrak{su}(N) \to \R
$$
be the symmetric bilinear form defined as $\langle X, Y \rangle = - \mathrm{tr}(XY)$ by using the trace of matrices. The restriction of the bilinear form $\langle \ , \ \rangle : \mathfrak{t} \times \mathfrak{t} \to \R$ is non-degenerate, which allows us to identify $\mathfrak{t}^*$ with $\mathfrak{t}$: For $\alpha \in \mathfrak{t}^*$, we define $\check{\alpha} \in \mathfrak{t}$ by $\alpha(\xi) = \langle \check{\alpha}, \xi \rangle$ for all $\xi \in \mathfrak{t}$. Also, we define $\langle \check{\alpha}, \check{\beta} \rangle = \langle \alpha, \beta \rangle$.

\medskip

The Weyl group $W$ is defined by $W = N(T)/T$, where $N(T) \subset SU(N)$ is the normalizer of the maximal torus. It acts on $T$ by conjugation, and induces actions on $\mathfrak{t}$ and $\mathfrak{t}^*$. The action of $W$ on $\mathfrak{t}$ preserves the inner product $\langle \ , \ \rangle$ as well as the integral lattice $\Pi = \mathrm{Ker}\exp 2\pi \subset \mathfrak{t}$, where $\exp 2\pi ( \cdot ) : \mathfrak{t} \to T$ is the exponential map. It is known that $W$ is isomorphic to the symmetric group $\mathfrak{S}_N$, which acts on $T$ and $\mathfrak{t}$ by permutation of diagonal components.

\medskip

For $i = 1, \ldots, N-1$, we define a homomorphism $\alpha_i : \mathfrak{t} \to \R$ by
$$
\alpha_i : \quad
\mathrm{diag}(\xi_1, \ldots, \xi_N) \mapsto 
\frac{\xi_i - \xi_{i+1}}{\ii},
$$
where $\xi_i \in \ii\R$ are subject to $\xi_1 + \cdots + \xi_N = 0$. These homomorphisms are called simple roots, and $\tilde{\alpha} = \alpha_1 + \cdots + \alpha_{N-1}$ is called the highest roots. For $i = 1, \ldots, N-1$, we also define the fundamental weights $\omega_i \in \mathfrak{t}^*$ by $\alpha_i = \sum_{j}A_{ij}\omega_j$, where $A_{ij} = 2\langle \alpha_i, \alpha_j \rangle/\langle \alpha_i, \alpha_i \rangle = 2\delta_{i, j} - \delta_{i, j+1} - \delta_{i+1, j}$. By definition, $\omega_i$ is characterized by  $\omega_i(\check{\alpha}_j) = \delta_{i, j}$. It holds that $$
\check{\alpha}_i =
\ii \, \mathrm{diag}(0, \cdots, 0, \overset{i}{1}, \overset{i+1}{-1}, 
0, \cdots, 0) \in \mathfrak{t},
$$
which shows $\Pi = \bigoplus_i \Z \check{\alpha}_i \subset \mathfrak{t}$. From this, we can see 
$$
\omega_i : \ 
\mathrm{diag}(\xi_1, \ldots, \xi_N) \mapsto
(\xi_1 + \cdots + \xi_i)/\ii,
$$
and $\omega_i$ form a basis of $\mathrm{Hom}(\Pi, \Z) \cong \mathrm{Hom}(T, U(1))$.

\medskip

Let $\mathfrak{C}$ and $\mathfrak{A}$ be subspaces in $\mathfrak{t}$ given by
\begin{align*}
\mathfrak{C} 
&=
\{ \xi \in \mathfrak{t} |\ 
\alpha_i(\xi) \ge 0, (i = 1, \ldots, N-1) \} \\
&=
\left\{ 
\sum_{i = 1}^{N-1} t_i \check{\omega}_i \in \mathfrak{t} \bigg|\ 
t_i \ge 0, (i = 1, \ldots, N-1) 
\right\}, \\
\mathfrak{A}
&=
\{ \xi \in \mathfrak{t} |\ 
\tilde{\alpha}(\xi) \le 1, \alpha_i(\xi) \ge 0, (i = 1, \ldots, N-1) \} \\
&=
\left\{ \sum_{i = 1}^{N-1} t_i \check{\omega}_i \in \mathfrak{t} \bigg|\ 
t_1 + \cdots + t_{N-1} \le 1, t_i \ge 0, (i = 1, \ldots, N-1) \right\}.
\end{align*}
which are called the (positive) Weyl chamber and alcove, respectively. Since the action of $W$ on $\mathfrak{t}$ is generated by the reflections with respect to the hyperplanes $H_i = \{ \xi \in \mathfrak{t} |\ \alpha_i(\xi) = 0 \}$, the Weyl chamber turns out to be a fundamental domain of $\mathfrak{t}$ with respect to the action of $W$. The alcove is also a fundamental domain of $\mathfrak{t}$ with respect to the affine Weyl group $\Pi \rtimes W$. Furthermore, $\exp 2\pi : \mathfrak{t} \to T \subset SU(N)$ induces a homeomorphism $\mathfrak{A} \cong SU(N)/SU(N)$ between the alcove and the space of conjugacy classes in $SU(N)$.


\subsubsection{Filtration of $SU(N)$\label{sec:filtration}}

\begin{defn}
	Let $N \ge 2$ be an integer.
	\begin{itemize}
		\item
		For $j = 1, \cdots, N-1$, we define an open subset $\mathfrak{A}_j$ in $\mathfrak{A}$ by
		\begin{align*}
		\mathfrak{A}_j &= 
		\{ \xi \in \mathfrak{t} |\ 
		\tilde{\alpha}(\xi) \le 1, 
		\alpha_j(\xi) > 0, \alpha_i(\xi) \ge 0 \ (i \neq j) \} \\
		&=
		\left\{ 
		\sum_{i = 1}^{N-1} t_i \check{\omega}_i \in \mathfrak{t} \bigg|\ 
		t_1 + \cdots + t_{N-1} \le 1, t_j > 0, t_i \ge 0, (i \neq j)
		\right\}.
		\end{align*}
		We also define an open subset $\mathfrak{A}_N$ in $\mathfrak{A}$ by
		\begin{align*}
		\mathfrak{A}_N &=
		\{ \xi \in \mathfrak{t} |\ 
		\tilde{\alpha}(\xi) < 1, \alpha_i(\xi) \ge 0 \ (i = 1, \cdots, N - 1 \} \\
		&=
		\left\{ 
		\sum_{i = 1}^{N-1} t_i \check{\omega}_i \in \mathfrak{t} \bigg|\ 
		t_1 + \cdots + t_{N-1} < 1, t_i \ge 0, \ (i = 1, \ldots, N-1) 
		\right\}.
		\end{align*}

		\item
		For $k = 0, \cdots, N-1$, we define an open subset $\mathfrak{A}_{\le k}$ in $\mathfrak{A}$ by
		$$
		\mathfrak{A}_{\le k} = 
		\bigcup_{1 \le j_1 < \cdots < j_{N - k} \le N }
		\mathfrak{A}_{j_1} \cap \cdots \cap \mathfrak{A}_{j_{N-k}}.
		$$
		
	\end{itemize}
\end{defn}

An element $\xi = \mathrm{diag}(\xi_1, \cdots, \xi_N) \in \mathfrak{t}$ consists of $\xi_i = \ii\lambda_i \in \ii\R$ such that $\lambda_1 + \cdots + \lambda_N = 0$. If $\xi \in \mathfrak{A}$, then $\xi_i = \ii\lambda_i$ further satisfies
$$
\lambda_1 \ge \cdots \ge \lambda_i \ge \lambda_{i+1} \ge \cdots \ge 
\lambda_N \ge \lambda_1 - 1.
$$
Thus, assuming that the $j$th inequality ``$\ge$'' in the above is strict ``$>$'', we get $\mathfrak{A}_j$ for $j = 1, \cdots, N-1$ and $\mathfrak{A}_N$
\begin{align*}
\mathfrak{A}_j
&=
\{ \ii\,\mathrm{diag}(\lambda_1, \cdots, \lambda_N) \in \mathfrak{t} |\
\lambda_1 \ge \cdots \ge \lambda_j > \lambda_{j+1} \ge 
\cdots \lambda_N \ge \lambda_1 - 1 
\}, \\
\mathfrak{A}_N
&=
\{ \ii\, \mathrm{diag}(\lambda_1, \cdots, \lambda_N) \in \mathfrak{t} |\
\lambda_1 \ge \cdots \ge \lambda_N > \lambda_1 - 1 \}.
\end{align*}
Similarly, assuming that there are $k$ non-strict inequalities ``$\ge$'', we get $\mathfrak{A}_{\le k}$. In particular, we have
\begin{align*}
\mathfrak{A}_{\le 0} &= 
\bigcap_{j = 1}^N \mathfrak{A}_j =
\{ \ii\,\mathrm{diag}(\lambda_1, \cdots, \lambda_N)  \in \mathfrak{t} |\
\lambda_1 > \cdots > \lambda_N > \lambda_1 - 1 \}, \\
\mathfrak{A}_{\le N-1} &= 
\bigcup_{j = 1}^N \mathfrak{A}_j =
\{ \ii\,\mathrm{diag}(\lambda_1, \cdots, \lambda_N)  \in \mathfrak{t} |\
\lambda_1 \ge \cdots \ge \lambda_N \ge \lambda_1 - 1 \} =
\mathfrak{A}.
\end{align*} 
There is clearly the following relation of inclusions
$$
\mathfrak{A}_{\le 0} \subset \cdots \subset
\mathfrak{A}_{\le k-1} \subset 
\mathfrak{A}_{\le k} \subset
\cdots \subset
\mathfrak{A}_{\le N-1} = \mathfrak{A}.
$$

\medskip

Let $q : SU(N) \to SU(N)/SU(N) \cong \mathfrak{A}$ be the quotient map. Using the open subsets in $\mathfrak{A}$ defined above, we make the following definition.

\begin{defn}
	Let $N \ge 2$ be an integer.
	\begin{itemize}
		\item
		For $j = 1, \cdots, N$, we define an open subset $O_j \subset SU(N)$ by $O_j = q^{-1}(\mathfrak{A}_j)$.

		\item
		For $k = 0, \cdots, N - 1$, we define an open subset $SU(N)_{\le k} \subset SU(N)$ by $SU(N)_{\le k} = q^{-1}(\mathfrak{A}_{\le k})$. We also define a closed subset $SU(N)_k$ in $SU(N)_{\le k}$ by $SU(N)_k =  SU(N)_{\le k} \backslash SU(N)_{\le k-1}$.
	\end{itemize}
\end{defn}

From the homeomorphism $SU(N)/SU(N) \cong \mathfrak{A}$, it follows that the eigenvalues of $U \in SU(N)$ are uniquely expressed as $\ee^{2\pi \ii \lambda_1(U)}, \cdots, \ee^{2\pi \ii \lambda_N(U)}$ by means of real numbers $\lambda_i(U) \in \R$ such that $\lambda_1(U) + \cdots + \lambda_N(U) = 0$ and 
$$
\lambda_1(U) \ge \cdots \ge \lambda_i(U) \ge \lambda_{i+1}(U) \ge \cdots
\ge \lambda_N(U) \ge \lambda_1(U) - 1.
$$
Then, we can express $O_j$ ($j = 1, \cdots, N-1$) and $O_N$ in $SU(N)$ as
\begin{align*}
O_j &=
\{ U \in SU(N) |\ 
\lambda_1(U) \ge \cdot\cdot \ge \lambda_j(U) > 
\lambda_{j+1}(U) \ge \cdot\cdot \ge 
\lambda_N(U) \ge \lambda_1(U) - 1 \}, \\
O_N &=
\{ U \in SU(N) |\ 
\lambda_1(U) \ge \cdots \ge \lambda_N(U) > \lambda_1(U) - 1 \},
\end{align*}
which form the same open cover $\{ O_j \}$ of $SU(N)$ as given in \cite{GawedzkiReis02,Meinrenken02}. The open sets $SU(N)_{\le k}$ have the relation of inclusions
$$
SU(N)_{\le 0} \subset \cdots \subset SU(N)_{\le k - 1} \subset
SU(N)_{\le k} \subset \cdots SU(N)_{\le N-1} = SU(N).
$$
Notice that $SU(N)_{\leq 1},\, SU(N)_1$ and $SU(N)_{\leq 0}$ have been already mentioned in \eqref{defSUN_<1}, \eqref{defSUN_1} and \eqref{defSUN_0}, respectively. Being an open subset of a manifold, $SU(N)_{\le k} \subset SU(N)$ is an open submanifold of dimension $\dim SU(N)_{\le k} = \dim SU(N) = N^2 - 1$.

\medskip

We write $\Delta_n$ for the $n$-dimensional simplex, and $\overset{\circ}{\Delta}_n \subset \Delta^n$ for its interior
\begin{align*}
\Delta_n &=
\{ (x_1, \cdots, x_{n+1}) \in \R^{n+1} |\ 
x_1 + \cdots + x_{n+1} = 1, x_i \ge 0 \ (i = 1, \cdots, n+1) \}, \\
\overset{\circ}{\Delta}_n &= 
\{ (x_1, \cdots, x_{n+1}) \in \R^{n+1} |\ 
x_1 + \cdots + x_{n+1} = 1, x_i > 0 \ (i = 1, \cdots, n+1) \}.
\end{align*}

\begin{lem} \label{lem:codimension_3}
	For $N \ge 2$, the closed subspace $SU(N)_1 \subset SU(N)_{\le 1}$ is a submanifold of codimension $3$. In particular, in the case of $N = 2$, the closed manifold consists of two points. In the case of $N > 2$, it is diffeomorphic to the disjoint union of $N$ copies of the $(N^2 - 4)$-dimensional manifold
	$$
	\overset{\circ}\Delta_{N-2} \times
	SU(N)/((U(2) \times 
	\overbrace{U(1) \times \cdots \times U(1)}^{N-2}) \cap SU(N)).
	$$
\end{lem}

\begin{proof}
	First of all, we identify $SU(N)_1$. We can express $\mathfrak{A}_{\le 1}$ as
	\begin{align*}
	&
	\bigcup_{j=1}^N 
	\mathfrak{A}_1 \cap \cdots \cap 
	\mathfrak{A}_{j-1} \cap \mathfrak{A}_{j+1} \cap 
	\cdots \cap \mathfrak{A}_N \\
	& =
	\bigcup_{j = 1}^N
	\{ \xi \in \mathfrak{t} |\
	\tilde{\alpha}(\xi) < 1, \alpha_j(\xi) \ge 0, \alpha_i(\xi) > 0 \
	(i \neq j) \} \\
	& =
	\bigcup_{j = 1}^N
	\{ \ii\,\mathrm{diag}(\lambda_1, \cdot\cdot, \lambda_N) \in \mathfrak{t} |\
	\lambda_1 > \cdots > \lambda_j \ge 
	\lambda_{j+1} > \cdots > \lambda_N > \lambda_1 - 1 \}.
	\end{align*}
	Therefore $\mathfrak{A}_{\le 1} \backslash \mathfrak{A}_{\le 0}$ is expressed as
	\begin{align*}
	&
	\bigcup_{j = 1}^N
	\{ \xi \in \mathfrak{t} |\
	\tilde{\alpha}(\xi) < 1, \alpha_j(\xi) = 0, \alpha_i(\xi) > 0 \
	(i \neq j) \} \\
	&=
	\bigcup_{j = 1}^N
	\{ \ii\,\mathrm{diag}(\lambda_1, \cdot\cdot, \lambda_N) \in \mathfrak{t} |\
	\lambda_1 > \cdot\cdot > \lambda_j = \lambda_{j+1} > 
	\cdot\cdot > \lambda_N > \lambda_1 - 1 \}.
	\end{align*}
	From this expression, if $N = 2$, then $\mathfrak{A}_{\le 1} \backslash \mathfrak{A}_{\le 0}$ is the two points that form the boundary of the $\mathfrak{A} \cong \Delta_{1}$. If $N > 2$, then $\mathfrak{A}_{\le 1} \backslash \mathfrak{A}_{\le 0}$ is the open $(N-2)$-dimensional face of the alcove $\mathfrak{A} \cong \Delta_{N-1}$, which is the disjoint union of $N$ copies of $\overset{\circ}{\Delta}_{N-2}$. Now, by definition, $SU(N)_1$ is the orbit of $\exp 2\pi (\mathfrak{A}_{\le 1} \backslash \mathfrak{A}_{\le 0}) \subset T$ under the adjoint action of $SU(N)$. It is known (see Section 3 in \cite{GawedzkiReis02} for example) that, for each $\xi \in \mathfrak{A}$, the stabilizer group (isotropy group) of $\xi \in \mathfrak{t}$ with respect to the adjoint action of $SU(N)$ agrees with that of $\exp 2\pi \xi \in T$. The stabilizer $SU(N)_\xi$ of $\xi = \ii\,\mathrm{diag}(\lambda_1, \cdots, \lambda_N) \in \mathfrak{A}_{\le 1} \backslash \mathfrak{A}_{\le 0}$ such that
	$$
	\lambda_1 > \cdots > \lambda_j = \lambda_{j+1} > 
	\cdots \lambda_N > \lambda_1 - 1
	$$
	is identified with 
	$$
	(U(1)^{j-1} \times U(2) \times U(1)^{N-1-j}) \cap SU(N),
	$$
	independent of $\xi$, provided that $\xi$ stays in the connected component. This concludes that the closed set $SU(N)_1 \subset SU(N)_{\le 1}$ is identified with the two point set in the case of $N = 2$, and, otherwise, with the disjoint union of $N$ copies of the manifold
	$$
	\overset{\circ}\Delta_{N-2} \times
	SU(N)/((U(2) \times 
	\overbrace{U(1) \times \cdots \times U(1)}^{N-2}) \cap SU(N)),
	$$
	whose dimension is computed as
	$$
	(N-2) + (N^2 - 1) - (2^2 + N - 2 - 1) = N^2 - 4.
	$$
	Let $f : S \to SU(N)_{\le 1}$ be the inclusion of
	$$
	S =
	\mathfrak{A}_{\le 1} \backslash \mathfrak{A}_{\le 0} \times
	SU(N)/((U(2) \times 
	\overbrace{U(1) \times \cdots \times U(1)}^{N-2}) \cap SU(N))
	$$
	onto $SU(N)_1$. This $f$ is a smooth map, since it can be constructed from the exponential map $\exp 2\pi : \mathfrak{A}_{\le 1} \backslash \mathfrak{A}_{\le 0} \to T \subset SU(N)$ and the adjoint action of $SU(N)$. One can describe a tangent vector on $S$ by means of the Lie algebra of $SU(N)$. Its image under the differential of $f$ is also described by means of the Lie algebra of $SU(N)$. This description helps us to see that $f$ is an immersion. It is clear that $f$ induces a homeomorphism from $S$ to its image $f(S) = SU(N)_1$ with the topology induced from $SU(N)_{\le 1}$. Therefore $f : S \to SU(N)_{\le 1}$ is an embedding, and its image $f(S) = SU(N)_1$ is a submanifold of $SU(N)_{\le 1}$. The codimension is $(N^2 - 1) - (N^2 - 4) = 3$.
\end{proof}

\begin{lem} \label{lem:codimension_higher}
	For $N \ge 3$ and $k = 2, \cdots, N-1$, the closed subspace $SU(N)_k \subset SU(N)_{\le k}$ is the disjoint union of submanifolds of codimension larger than $3$.
\end{lem}

\begin{proof}
	The argument of the proof is a straight generalization of that of Lemma~\ref{lem:codimension_3}. In the case that $1 \le k \le N-2$, the closed set $\mathfrak{A}_{\le k} \backslash \mathfrak{A}_{\le k-1}$ is the disjoint union of $\binom{N}{k-1}$ copies of the $(N-k-1)$-dimensional open simplex $\overset{\circ}{\Delta}_{N-K-1}$. In the case that $k = N - 1$, the closed set $\mathfrak{A}_{\le N-1} \backslash \mathfrak{A}_{\le N - 2}$ consists of $N$ points. Recall that an element $\xi = \ii\,\mathrm{diag}(\lambda_1, \cdots, \lambda_N) \in \mathfrak{A}$ is such that
	$$
	\lambda_1 \ge \lambda_2 \ge \cdots \ge \lambda_N \ge \lambda_1 - 1.
	$$
	For $k = 1, \cdots, N-1$, we have $\xi \in \mathfrak{A}_{\le k} \backslash \mathfrak{A}_{\le k-1}$ if and only if there are $k$ equalities ``$=$'' and $(N-k)$ strict inequalities ``$>$'' among $N$ inequalities ``$\ge$'' in the above. The closed subspace $SU(N)_k$ in $SU(N)_{\le k}$ is the orbit of $\exp 2\pi (\mathfrak{A}_{\le k} \backslash \mathfrak{A}_{\le k-1})$ under the adjoint action of $SU(N)$. For $\xi \in \mathfrak{A}_{\le k} \backslash \mathfrak{A}_{\le k-1}$, the stabilizer of $\exp 2\pi \xi \in T$ agrees with that of $\xi \in \mathfrak{t}$, and is isomorphic to a group of the form
	$$
	(U(M_1) \times \cdots \times U(M_{N-k})) \cap SU(N),
	$$
	where $M_1, \cdots, M_{N-k}$ are positive integers such that $M_1 + \cdots + M_{N-k} = N$. Therefore $SU(N)_k \subset SU(N)_{\le k}$ can be identified with the disjoint union of manifolds of dimension
	$$
	(N - k - 1) + (N^2- M_1^2 - \cdots - M_{N-k}^2)
	= (N^2 - 1) - (M_1^2 + \cdots + M_{N-k}^2 - N + k).
	$$
	In the same way as in the proof of Lemma~\ref{lem:codimension_3}, we can show that $SU(N)_k \subset SU(N)_{\le k}$ is the disjoint union of the submanifolds. Their codimensions in $SU(N)_{\le k}$ are $M_1^2 + \cdots + M_{N-k}^2 - N + k$. Once this codimension is shown to be larger than $3$, the proof of the lemma will be completed. This claim about the codimension can be shown as follows: To suppress notations, let us put $\ell = N - k$. Then $2 \le k \le N - 1$ if and only if $1 \le \ell \le N - 2$. Applying the Cauchy-Schwartz inequality and $\ell \le N - 2$, we get
	\begin{align*}
	M_1^2 + \cdots + M_\ell^2 - \ell
	&= (M_1^2 + \cdots + M_\ell^2)
	(\overbrace{1^2 + \cdots + 1^2}^\ell) \cdot \frac{1}{\ell} - \ell \\
	&\ge
	(M_1 \cdot 1 + \cdots + M_\ell \cdot 1)^2 \cdot \frac{1}{\ell} - \ell \\
	&=
	\frac{N^2}{\ell} - \frac{\ell^2}{\ell}
	= \frac{(N - \ell)(N + \ell)}{\ell}
	= (N - \ell)\bigg( \frac{N}{\ell} + 1 \bigg) \\
	&\ge
	2\bigg( \frac{\ell + 2}{\ell} + 1 \bigg)
	= 2 \bigg( \frac{2}{\ell} + 2 \bigg) \\
	&> 2 \cdot (0 + 2) = 4,
	\end{align*}
	as claimed.
\end{proof}


\subsubsection{Proof of Proposition~\ref{prop:ensure_general_position}}

As is known, any continuous map $U : X \to SU(N)$ is homotopic to a smooth map $U^{(N-1)} : X \to SU(N) = SU(N)_{\le N-1}$. For $j = 1, \cdots, N-1$, suppose that we have a smooth map $U^{(N-j)} : X \to SU(N)_{\le N - j}$. Since $SU(N)_{\le N-j}$ and its submanifolds constituting $SU(N)_{N-j}$ have no boundary, we can apply the transversality homotopy theorem \cite{GuilleminPolack10} to $U^{(N-j)}$, so that $U^{(N-j)}$ is homotopic to a smooth map $U^{(N-j-1)} : X \to SU(N)_{\le N -j}$ which is transverse to each component of $SU(N)_{N-j} \subset SU(N)_{\le N-j}$. For $j = 1, \cdots, N-2$, the sum of the dimension of $X$ and the dimension of each manifold constituting $SU(N)_{N-j}$ is less than the dimension of $SU(N)$ by Lemma~\ref{lem:codimension_higher}. Thus, in this case, the transversality means that the image $U^{(N-j-1)}(X)$ has no intersection with $SU(N)_{N - j}$, so that we can regard $U^{(N-j-1)}$ as a smooth map $U^{(N-j-1)} : X \to SU(N)_{\le N - j - 1}$. As a result, an induction shows that $U : X \to SU(N)$ is homotopic to a smooth map $U^{(0)} : X \to SU(N)$ such that the image $U^{(0)}(X)$ of $X$ is contained in $SU(N)_{\le 1}$ and intersects with $SU(N)_1$ transversally. Now, the dimension of $X$ and that of $SU(N)_1$ add up to the dimension of $SU(N)$ by Lemma~\ref{lem:codimension_3}. Then a consequence of the transversality is that the inverse image $(U^{(0)})^{-1}(SU(N)_1) \subset X$ is a submanifold of dimension $0$. Since $X$ is compact, the inverse image consists of a finite number of points. This completes the proof that $U$ is homotopic to a smooth map $U' = U^{(0)}$ with the required property.

Note that the assumption $\mathrm{dim}X = 3$ is necessary for $U'$ to have only single eigenvalue crossings of $U$ or less. When the dimension of $X$ is larger than $3$, there are generally multiple crossings of eigenvalues.

\begin{rem}\label{rem:deformation_boundary}
	We can generalize Proposition~\ref{prop:ensure_general_position} to the case where $X$ has a non-empty boundary $\partial X$. In this case, we additionally assume that $U|_{\partial X}$ is smooth and $U(\partial X) \subset SU(N)_{\le 1}$. Then $U$ is homotopic to $U' : X \to SU(N)$ such that
	\begin{itemize}
		\item
		the image $U'(X)$ of $X$ under $U'$ is contained in $SU(N)_{\le 1} \subset SU(N)$, and
		
		\item
		the inverse image ${U'}^{-1}(SU(N)_1)$ of $SU(N)_1$ under $U'$ consists of a finite number of points, and
		
		\item
		$U'|_{\partial X} = U|_{\partial X}$,
	\end{itemize}
	under a homotopy $\tilde{U} : X \times [0, 1] \to SU(N)$ such that $\tilde{U}|_{\partial X \times [0, 1]} = U$. 
\end{rem}

\subsection{Local formula\label{sec:proof_mainthm}}

\subsubsection{Basic gerbe data}
Recall the open sets $O_1, \cdots, O_N$ of $SU(N)$. They constitute an open cover $\{ O_j \}$ of $SU(N)$. For $j, k$ such that $1 \le j < k < N$, the intersection $O_j \cap O_k$ consists of $U \in SU(N)$ whose eigenvalues are expressed as $\ee^{2\pi \ii \lambda_1(U)}, \cdots, \ee^{2\pi \ii \lambda_N(U)}$ in terms of $\lambda_i \in \R$ satisfying $\lambda_1(U) + \cdots + \lambda_N(U) = 0$ and 
$$
\lambda_1(U) \ge \cdot\cdot \ge \lambda_j(U) >
\lambda_{j+1}(U) \ge \cdot\cdot \ge \lambda_k(U) >
\lambda_{k+1} \ge \cdot\cdot \lambda_N(U) \ge \lambda_1(U) - 1.
$$
For $j = 1, \cdots, N-1$, the intersection $O_j \cap O_N$ consists of $U \in SU(N)$ such that
$$
\lambda_1(U) \ge \cdots \ge \lambda_j(U) >
\lambda_{j+1}(U) \ge \cdots \ge \lambda_N(U) > \lambda_1(U) - 1.
$$
Thus, if $1 \le j < k \le N$, then the eigenvalues $\ee^{2\pi \ii \lambda_{j+1}(U)}, \cdots, \ee^{2\pi \ii \lambda_{k}(U)}$ of $U \in O_j \cap O_k$ are separated by the remaining ones by ``gaps'', so that their eigenvectors constitute a complex vector bundle $E_{jk} \to O_j \cap O_k$ of rank $(k - j)$. We define a complex line bundle $L_{jk} \to O_j \cap O_k$ to be the determinant line bundle (the top exterior product): $L_{jk} = \det E_{jk}$. In the following we will apply the convention $L_{jk} = L_{kj}^*$ if $j > k$.

\begin{prop}[\cite{GawedzkiReis02,Meinrenken02}]
	Let $N \ge 2$ be an integer. There exist:
	\begin{itemize}
		\item
		differential $2$-forms $B_j \in \Omega(O_j)$ such that $dB_j = H$ on $O_j$, where $H \in \Omega^3(SU(N))$ is the $3$-form
		$$
		H = \frac{1}{24\pi^2} \mathrm{tr}(g^{-1}dg),
		$$
		which represents the integral image of $W_3 \in H^3(SU(N); \Z)$; and

		\item
		connections $A_{jk}$ on $L_{jk}$ such that $B_k - B_j = c_1(A_{jk})$ on $O_j \cap O_k$, where $c_1(A_{jk})$ is the first Chern form associated to $A_{jk}$, which represents the integral image of the first Chern class $c_1(L_{jk}) \in H^2(O_j \cap O_k; \Z)$.
	\end{itemize}
\end{prop}

The proposition above follows from the constructions in \cite{GawedzkiReis02,Meinrenken02}: We will not enter into its detail here, but the line bundles $L_{jk}$ constitute a part of the data of the basic gerbe on $SU(N)$. The differential $2$-forms $B_j$ and the connections $A_{jk}$ are the data of a connection on the basic gerbe.


The line bundle $\mathcal L_a \to \partial N_a$ and $\Ch(X_a;j)$ from Lemma~\ref{def:Ch_X} admit an expression in terms of this data. For convenience, put $\bar{j} = j - 1$ for $j = 2, \cdots, N$ and $\bar{1} = N$. We find
$$
U(\partial N_a) 
\subset O_{\bar{j}} \cap O_j,
$$
the line bundle $\mathcal{L}_a \to N_a$ is the pull-back of $L_{\bar{j}j} \to O_{\bar{j}} \cap O_j$
$$
\mathcal{L}_a = U|_{\partial N_a}^* L_{\bar{j}j}.
$$
and 
$$
\Ch(X_a; j) = \int_{\partial N_a} U^*c_1(A_{\bar{j}j}).
$$
In the definition above, $\partial N_a$ inherits an orientation from $X$. It is clear that $\Ch(X_a; j)$ is independent of the choice of $N_a$. Moreover, if $X_a = \{x\}$ then $N_a = D_x$ and the latter expressions naturally apply to $\mathcal L_x \to \partial D_x$ and $\Ch(x;j)$ from Definition~\ref{def:Ch_x}. Finally, as already pointed out in Remark~\ref{rk:submanifolds}, if $N_a \cong X_a \times [-1,1]$ then
\begin{align*}
\Ch(X_a; j)
&= - \int_{X_a} U|_{X_a \times \{ 1 \}}^*c_1(\mathcal{L}_{\bar{j}j})
+ \int_{X_a} U|_{X_a \times \{ -1 \}}^*c_1(\mathcal{L}_{\bar{j}j}).
\end{align*}

Similarly, for $U^{-1}(\1) = \bigsqcup_b Y_b$, we can find a closed tubular neighborhood $N_b$ of each $Y_b$ such that: the eigenvalues of $U(y)$ are distinct for any $y \in N_b \backslash Y_b$; and $N_b \cap N_{b'} = \emptyset$ whenever $b \neq b'$. However, the identity matrix $U = \1 = \mathrm{diag}(1, \cdots, 1) \in SU(N)_{N-1}$ reads
$$
\overbrace{\lambda_1(\1)}^{0} = \cdots = 
\overbrace{\lambda_N(\1)}^{0} > 
\overbrace{\lambda_1(\1)}^{-1}.
$$
Therefore $\1 \in O_N$ and $\1 \notin O_j$ for $j\neq N$. Moreover
$
U(\partial N_b) \subset O_N \cap O_j
$
so that
$
\mathcal{L}_b = U|_{\partial N_b}^* L_{Nj}
$
and 
$$
\int_{\partial N_b} U^*c_1(A_{Nj})= -\sum_{\ell=j+1}^{N} \Ch(Y_b;j)
$$
which vanishes for $j=N$.

\subsubsection{Proof of the local formula}

The proof of Theorem~\ref{thm:main} is analogous to the one Theorem~\ref{thm:full_degeneracy_at_1} if we set $U^{-1}(\1) = \emptyset$ and $\dim X_a=0$ for all $X_a \subset Cr_j$. Thus we focus on Theorem~\ref{thm:full_degeneracy_at_1}. 

We write $X'$ for the closure of the complement of $\bigsqcup_{a} N_a \sqcup \bigsqcup_b N_b \subset X$, where $N_a$ and $N_b$ are the disjoint closed tubular neighborhoods used in the definitions of $\Ch(X_a; j)$ and $\Ch(Y_b; j)$. The boundary of $X'$ is the disjoint union of $\partial N_a$ and $\partial N_b$ with the opposite orientation. Note that $U(N_a) \subset O_{\bar{j}}$, $U(N_b) \subset O_N$ and $U(X') \subset O_j$. Hence the image $U(X)$ of $X$ under $U$ is covered by $O_{\bar{j}}$, $O_N$ and $O_j$. Now, we use Stokes' theorem to get
\begin{align*}
W_3(U) &= \int_X U^*H 
= \sum_{a} \int_{N_a} U^*H +  \sum_b \int_{N_b} U^* H + \int_{X'} U^*H \\
&= \sum_{a} \int_{N_a} U^* dB_{\bar{j}} + \sum_b\int_{N_b} U^*dB_N 
+ \int_{X'} U^*dB_j \\
&= \sum_{a} \int_{\partial N_a} U^*B_{\bar{j}}
+ \sum_b \int_{\partial N_b} U^*B_N 
- \sum_{a} \int_{\partial N_a} U^*B_j
- \sum_b \int_{\partial N_b} U^*B_j \\
&= - \sum_{a} \int_{\partial N_a} U^*c_1(A_{\bar{j}j})
- \sum_b \int_{\partial N_b} U^*c_1(A_{Nj}),
\end{align*}
which leads to the formulae in the theorem.

\begin{rem}
For $X_a \subset \Cr_j(U)$  the local Chern number is defined for the $j$th eigenvector line bundle but we could in principle consider the others, such as the determinant line bundle of the rank $2$ vector bundle whose fiber at $y \in \partial N_a$ is spanned by the eigenvectors with eigenvalues $\ee^{2\pi \ii \lambda_{j-1}(U(y))}$ and $\ee^{2\pi \ii \lambda_j(U(y))}$. However, these eigenvalues remain distinct at any $y \in N_a$, including $X_a$, so that the line bundle whose fiber at $y \in \partial N_a$  is associated to $\ee^{2\pi \ii \lambda_{j-1}(U(y))}$ is trivial. Therefore the Chern number of the determinant line bundle agrees with  $\Ch(X_a;j)$ from Lemma~\ref{def:Ch_X}. Put differently, the $j$th eigenvector is essential to the local Chern number of the $j$th eigenvalue crossing and the others are not.

On the other hand such simplification does not occur for $Y_b \subset U^{-1}(\1)$ because $1 \notin O_j$ for $j=1,\ldots, N-1$, so that in general one has to consider all eigenvectors between $j+1$ and $N$ for the $j$th crossing.
\end{rem}

\subsection{Floquet map classification\label{sec:proof_floquet}}

In this section we prove Theorem~\ref{thm:Floquet_maps} that contains three equivalent statements. We start by proving the equivalence between 2 and 3, that is a direct consequence of Proposition~\ref{prop:rel_I_C}, and then show the equivalence between 1 and 3.

\subsubsection{Proof of Proposition~\ref{prop:rel_I_C}}

By Lemma~\ref{label_ev} we label the eigenvalues of $V=\Sigma \times \{1\} \to SU(N)_{\leq 0}$ by $\lambda_1 > \lambda_2 > \ldots > \lambda_N > \lambda_1 -1$ with $\lambda_1 + \ldots + \lambda_N =0$. Then it exists $\alpha_i : \Sigma \rightarrow \mathbb R$ for $i=1,\ldots,N$ continuous (or even smooth) such that $\lambda_i > \alpha_i > \lambda_{i+1}$ for $1\leq i \leq N-1$ and $\lambda_N +1 > \alpha_N > \lambda_1$. We define the complex logarithm with branch cut $\ee^{2\pi \ii \alpha}$ for $\alpha \in \mathbb R$ by 
\begin{equation}\label{deflog}
\log_\alpha(\ee^{2 \pi \ii \phi}) = 2\pi \ii \phi \qquad \mathrm{if}  \qquad \alpha-1 < \phi < \alpha.
\end{equation}
We can then define the effective Hamiltonian $H_{\alpha_j} = \tfrac{1}{2\pi \ii} \log_{\alpha_j}(V) : \Sigma \rightarrow M_N(\mathbb C)$ and 
\begin{equation*}
U_{\alpha_j}(k,t) = \left\lbrace \begin{array}{ll}
U(k,t), & 0\leq t\leq 1 \\
\ee^{2\pi \ii (2-t) H_{\alpha_j(k)}(k)}, & 1\leq t\leq 2.
\end{array} \right.
\end{equation*}
This relative evolution is periodic in time. The proof of Proposition~\ref{prop:rel_I_C} follows from the following Lemma.
\begin{lem}
	For any $j,p=1,\ldots,N$ one has
	\begin{equation}\label{W3_relative}
	W_3(U_{\alpha_j}) = \left\lbrace \begin{array}{ll} \mathcal I(U;p), &  p=j\\
	 \mathcal I(U;p) + \sum_{\ell=p+1}^j C(V,\ell), & j<p\cr \mathcal I(U;p) - \sum_{\ell=j+1}^p C(V,\ell), & j>p.\end{array}\right.
	\end{equation}
\end{lem}

\begin{proof}
For a given $j=1,\ldots,N$ the eigenvalues of $H_{\alpha_j}$ are $\mu_i^j : \Sigma \rightarrow  \mathbb R$ for $i=1,\ldots,N$ with
\begin{equation*}
\qquad \mu_i^j(k) := \dfrac{1}{2\pi \ii}\log_{\alpha_j(k)}(\ee^{2 \pi \ii \lambda_i(k)})  
\end{equation*}
from \eqref{deflog} we infer
\begin{align}\label{muvslambda}
&\mu_i^j =  \lambda_i,\qquad j+1\leq i \leq N\cr
 &\mu_i^j =\lambda_i -1 ,\qquad 1\leq i \leq j
\end{align}
and
\begin{equation*}
	\mu_i^N = \lambda_i \qquad \forall i
\end{equation*}
This implies 
\begin{equation*}
\mu_{j+1}^j > \mu_{j+2}^j > \ldots \mu_N^j > \mu_1^j > \ldots \mu_j^j > \mu_{j+1}^j -1
\end{equation*}
for $1\leq j \leq N-1$ and $\mu_1^N > \mu_2^N > \ldots > \mu_N^N > \mu_1 -1$ for $j=N$ so that the logarithm only rearranges the eigenvalues. However
\begin{align*}
&\sum_{i=1}^N \mu_i^j = -j \neq 0, \qquad 1\leq j \leq N-1, \\
&\sum_{i=1}^N \mu_i^N =0.
\end{align*}
Notice that $\tr( s H_{\alpha_j}) = -s j$ where $s=(2-t) \in [0,1]$, except for $j=N$ where $\tr( s H_{\alpha_N}) =0$ so that $U_{\alpha_j} : \Sigma \to U(N)$ for $1 \leq j \leq N-1$ but instead $U_{\alpha_N} : \Sigma \to SU(N)$. So we start with $j=N$ where Theorem~\ref{thm:full_degeneracy_at_1} applies. By construction, the crossings of $U_{\alpha_N}$ occur only for $t \in [0,1]$ and at $t=2$ where $U_{\alpha_N}(\cdot,2)=\1$. Indeed the eigenvalues of $U_{\alpha_N}$ for $t \in (1,2)$ are 
\begin{equation*}
s \lambda_1 > \ldots > s \lambda_N > s \lambda_1-1
\end{equation*}
where $s=2-t \in (0,1)$ and $\lambda_i$ are the eigenvalue of $V = U_{\alpha_N}(\cdot,1)$. Thus
\begin{align*}
& W_3(U_{\alpha_N}) = - \sum_a \Ch(X_a,p) + \sum_b \sum_{\ell=p+1}^N \Ch(Y_b;\ell), \qquad 1\leq p \leq N-1, \\
& W_3(U_{\alpha_N}) = - \sum_a \Ch(X_a,N).
\end{align*}
In the sum over $b$ we only have one piece where $Y_b = \Sigma \times \{0\}  = \Sigma \times \{2\}$. A tubular neighborhood consists in two pieces: a collar neighborhood of $\Sigma \times \{0\}$ that appears in  Definition~\ref{def:top_floquet} of $\mathcal I(U;j)$; and a collar neighborhood of $\Sigma \times \{2\}$, that we take as $\Sigma \times \{1\}$ so that we get
\begin{align}
& W_3(U_{\alpha_N}) = \mathcal I(U;p) + \sum_{\ell=p+1}^N C(V;\ell)\label{W3_U_alpha_N_Topj}\\
& W_3(U_{\alpha_N}) = \mathcal I(U;N) 
\end{align}
which implies
\begin{equation*}
\mathcal I(U;N) - \mathcal I(U;p) = \sum_{\ell=p+1}^N C(V;\ell)
\end{equation*}

Finally let's come back to $U_{\alpha_j}$ for $1\leq j \leq N-1$. This map is not $SU(N)$-valued anymore so Theorem~\ref{thm:full_degeneracy_at_1} does not apply. However from \eqref{muvslambda} one has
\begin{equation*}
H_{\alpha_j} = H_{\alpha_N} - P_{1,j}
\end{equation*}
where $P_{1,j}$ is the eigenprojection associated to the eigenvalues $\lambda_1, \ldots, \lambda_j$ of $V$. We deduce 
\begin{equation*}
U_{\alpha_j} = U_{\alpha_N} U_{1,j}, \qquad  U_{1,j} = \left\lbrace \begin{array}{ll}
1, & 0\leq t\leq 1 \\
\ee^{2\pi \ii (t-2) P_{1,j}}, & 1\leq t\leq 2.
\end{array} \right.
\end{equation*}
Both $U_{\alpha_N}$ and $U_{1,j}$ are time periodic so that $W_3(U_{\alpha_j} ) = W_3(U_{\alpha_N} ) + W_3(U_{1,j})$. A standard computation shows (see \cite{Rudner13}) that $W_3(U_{1,j}) = \sum_{\ell=1}^j C(V;\ell)$. Thus from \eqref{W3_U_alpha_N_Topj} we get
\begin{equation*}
	W_3(U_{\alpha_j} ) = \mathcal I(U;p) + \sum_{\ell=p+1}^N C(V;\ell)  + \sum_{\ell=1}^j C(V;j) 
\end{equation*}
For $p = j$ this is $W_3(U_{\alpha_j} ) = \mathcal I(U;j)$ as the total Chern vanishes. For $p \neq j$ we get \eqref{W3_relative}.
\end{proof}

\subsubsection{Proof of Theorem~\ref{thm:Floquet_maps}}

We start by the case of two Floquet maps that coincide at their endpoints, and prove Proposition~\ref{prop:Floquet_maps}. We then deal with the general case. Once the extra Chern numbers $C(V,\ell)$ are taken into account, the proof of Theorem~\ref{thm:Floquet_maps} relies on Proposition~\ref{prop:Floquet_maps}.

\begin{proof}[Proof of Proposition~\ref{prop:Floquet_maps}] 
We can think of $U_0U_1^{-1} : \Sigma \times [0, 1] \to SU(N)$ as $U_0U_1^{-1} : \Sigma \times S^1 \to SU(N)$, since $U_0(x, i)U_1(x, i)^{-1} = \1$ for $x \in \Sigma$ and $i = 0, 1$. We write $\hat U$ for $U_0U_1^{-1}$ regarded as a map $\Sigma \times S^1 \to SU(N)$. Proposition~\ref{prop:obstruction} shows that $U_0$ and $U_1$ are homotopic relative to $\Sigma \times \partial [0, 1]$ if and only if $W_3(\hat U) = 0$. So we prove $W_3(\hat U) = \Top(U_0; j) - \Top(U_1; j)$. For this aim, let us consider $U_0 \cup U_1 : \Sigma \times S^1 \to SU(N)$ as the gluing of $\Sigma \times [0, 1]$ and its copy with the opposite orientation along their boundaries. Then Theorem~\ref{thm:full_degeneracy_at_1} implies $W_3(U_0 \cup U_1) = \Top(U_0; j) - \Top(U_1; j)$. Since $W_3(U_1 \cup U_1) = 0$, we have 
\begin{align*}
W_3(U_0 \cup U_1) &= W_3(U_0 \cup U_1) - W_3(U_1 \cup U_1)  \cr & = W_3((U_0 \cup U_1)(U_1 \cup U_1)^{-1}) \cr &= W_3((U_0U_1^{-1}) \cup \1).
\end{align*}
Regarding the maps $U_0U_1^{-1}$ and $\1$ from $\Sigma \times [0, 1]$ as maps $\hat U$ and $\1$ from $\Sigma \times S^1$, we have $W_3((U_0U_1^{-1}) \cup \1) = W_3(\hat U) + W_3(\1)$. Because $W_3(\1) = 0$, we get $W_3(U_0 \cup U_1) = W_3(\hat U)$, and hence $W_3(\hat U) = \Top(U_0; j) - \Top(U_1; j)$.
\end{proof}

\begin{lem}\label{lem:hom_V_C}
Let $V_i : \Sigma \to SU(N)_{\le 0}$, ($i = 0, 1$) be given. Then, $V_0$ and $V_1$ are homotopic (as maps with values in $SU(N)_{\le 0}$) if and only if $C(V_0; p) = C(V_1; p)$ for all $p = 1, \ldots, N$.
\end{lem}
\begin{proof}
The open submanifold $SU(N)_{\le 0}$ in $SU(N)$ is homotopy equivalent to $SU(N)/T$. The exact sequence of homotopy groups associated to the fibration $T \to SU(N) \to SU(N)/T$ allows us to compute the homotopy groups of $SU(N)/T$ as follows
$$
\pi_n(SU(N)/T)
\cong
\left\{
\begin{array}{ll}
0, & (n = 0, 1) \\
\pi_1(T) \cong \Z^{N-1}, & (n = 2) \\
\pi_N(SU(N)). & (n \ge 3)
\end{array}
\right.
$$
Then, generalizing the argument in Proposition~\ref{prop:obstruction}, we can see that the homotopy group $\pi_2(SU(N)/T)$ obstructs the existence of a homotopy between two maps $V, V' : \Sigma \to SU(N)_{\le 0}$. The obstruction is then identified with the differences $C(V_0; p) - C(V_1; p)$ for $p = 1, \ldots, N-1$. 
\end{proof}

\paragraph{1 $\Rightarrow$ 3:} Suppose there exists such a homotopy $\tilde{U}$ as stated. This homotopy restricts to a homotopy $\tilde{U}|_{\Sigma \times \{ 1 \} \times [0, 1]}$ between $U_0|_{\Sigma \times \{ 1 \}}$ and $U_1|_{\Sigma \times \{ 1 \}}$. Hence $C(U_0|_{\Sigma \times \{ 1 \}}; p) = C(U_1|_{\Sigma \times \{ 1 \}}; p)$ for all $p$ by Lemma~\ref{lem:hom_V_C}. Gluing $U_0$ and the homotopy $\tilde{U}|_{\Sigma \times \{ 1 \} \times [0, 1]}$, we define $U'_0 : \Sigma \times [0, 1] \to SU(N)$ as follows
$$
U'_0(x, t) =
\left\{
\begin{array}{ll}
U_0(x, 2t), & ((x, t) \in \Sigma \times [0, 1/2]) \\
\tilde{U}(x, 1, 2t - 1). & ((x, t) \in \Sigma \times [1/2, 1])
\end{array}
\right.
$$
We have $U'_0|_{\Sigma \times \{ 1 \}} = \tilde{U}|_{\Sigma \times \{ 1 \} \times \{ 1 \}} = U_1|_{\Sigma \times \{ 1 \}}$ by assumption. Because $\tilde{U}(\Sigma \times \{ 1 \} \times [0, 1]) \subset SU(N)_{\le 0}$, we see $\Top(U'_0; j) = \Top(U_0; j)$. Now, we define $\tilde{U}' : \Sigma \times [0, 1] \times [0, 1] \to SU(N)$ by
$$
\tilde{U}'(x, t, s)
=
\left\{
\begin{array}{ll}
\tilde{U}(x, 2t, s), & ((x, t) \in \Sigma \times [0, (1-s)/2]) \\
\tilde{U}(x, \frac{2st + 1 - s}{1+s}, \frac{(2-2s)t + 3s-1}{1+s}). & 
((x, t) \in \Sigma \times [(1-s)/2, 1])
\end{array}
\right.
$$
This is a homotopy between $U'_0$ and $U_1$ relative to $\Sigma \times \partial [0, 1]$. Thus, Proposition~\ref{prop:Floquet_maps} implies $\Top(U'_0; j) = \Top(U_1; j)$, and hence $\Top(U_0; j) = \Top(U_1; j)$. 

\paragraph{3 $\Rightarrow$ 1:} If $C(V_0;p) = C(V_1;p)$ for all $p$, then $U_0|_{\Sigma \times \{ 1 \}}$ and $U_1|_{\Sigma \times \{ 1 \}}$ are homotopic as maps to $SU(N)_{\le 0}$ by Lemma~\ref{lem:hom_V_C}. Let $\tilde{V} : \Sigma \times \{ 1 \} \times [0, 1] \to SU(N)_{\le 0}$ be such a homotopy with $\tilde{V}|_{\Sigma \times \{ 1 \} \times \{ i \}} = U_i|_{\Sigma \times \{ 1 \}}$. We define $W_0 : \Sigma \times [0, 1] \to SU(N)$ by concatenation of $U_0$ and $\tilde{V}$, 
$$
W_0(x, t) =
\left\{
\begin{array}{ll}
U_0(x, 2t), & ((x, t) \in \Sigma \times [0, 1/2]) \\
\tilde{V}(x, 1, 2t-1). & ((x, t) \in \Sigma \times [1/2, 1])
\end{array}
\right.
$$
If we define $\tilde{W}_0 : \Sigma \times [0, 1] \times [0, 1] \to SU(N)$ by
$$
\tilde{W}_0(x, t, s) 
=
\left\{
\begin{array}{ll}
U_0(x, 2t/(1 + s)), & ((x, t) \in \Sigma \times [0, (1 + s)/2]) \\
\tilde{V}(x, 1, 2t - 1 - s), & ((x, t) \in \Sigma \times [(1+s)/2, 1])
\end{array}
\right.
$$
then $\tilde{W}_0$ is a homotopy between $W_0$ and $U_0$ such that $\tilde{W}|_{\Sigma \times \{ 0 \} \times [0, 1]} = \1$ and $\tilde{W}(\Sigma \times \{ 1 \} \times [0, 1]) \subset SU(N)_{\le 0}$. Thus, the proposition will be completed by showing that $W_0$ and $W_1 = U_1$ are homotopic relative to $\Sigma \times \partial [0, 1]$. Note that $W_0$ satisfies the assumptions for $\Top(W_0; j)$ to be defined. Since $\tilde{V}$ is a homotopy in $SU(N)_{\le 0}$, we have $\Top(W_0; j) = \Top(U_0; j)$. Now, by Proposition~\ref{prop:Floquet_maps}, there is a homotopy between $W_0$ and $W_1 = U_1$ relative to $\Sigma \times \partial [0, 1]$. 

\appendix

\section{Reduction to $SU(N)$-valued maps}
\label{sec:reduction}

Let $X$ be a topological space, and $Y \subset X$ a subspace. For a topological group $G$, we denote by $C((X, Y), (G, 1))$ the set of continuous maps $U : X \to G$ such that $U|_Y \equiv 1$ is the constant map at the unit $1 \in G$. By the pointwise multiplication, the set gives rise to a group. A (relative) homotopy between two maps $U_0, U_1 \in C((X, Y), (G, 1))$ is a continuous map $\tilde{U} \in C((X \times [0, 1], Y \times [0, 1]), (G, 1))$ such that $\tilde{U}|_{X \times \{ i \}} = U_i$ for $i = 0, 1$. The set of homotopy classes in $C((X, Y), (G, 1))$ will be denoted by
$$
[(X, Y), (G, 1)],
$$
which inherits a group structure from $C((X, Y), (G, 1))$.

\begin{lem} \label{lem:exact_sequence}
Let $X$ be a topological space, and $Y \subset X$ a subspace. There is an exact sequence of groups
$$
1 \to 
[(X, Y), (SU(N), 1)] \to
[(X, Y), (U(N), 1)] \to
[(X, Y), (U(1), 1)] \to
1.
$$
This admits a section to the surjection induced from $\det : U(N) \to U(1)$, so that there is an isomorphism of groups
$$
[(X, Y), (U(N), 1)]
\cong [(X, Y), (SU(N), 1)] \rtimes [(X, Y), (U(1), 1)].
$$
\end{lem}

\begin{proof}
We have the exact sequence of topological groups
$$
1 \to SU(N) \to U(N) \overset{det}{\to} U(1) \to 1,
$$
which admits a section $s : U(1) \to U(N)$ given by $s(u) = \mathrm{diag}(u, 1, \cdots, 1)$. Using this section, we can verify the lemma directly.
\end{proof}

To describe the obstructions for $U \in C((X, Y), (U(1), 1))$ to being homotopic to the constant map at $1$, we introduce the odd dimensional winding number as follows: It is well known that the cohomology ring $H^*(U(N); \Z)$ of $U(N)$ is isomorphic to the exterior ring 
$$
H^*(U(N); \Z) \cong \bigwedge (W_1, W_3, \cdots, W_{2N-1})
$$
generated by $W_{2i-1} \in H^{2i-1}(U(N); \Z) \cong H^{2i-1}(U(N), 1; \Z)$, ($i = 1, \cdots, N$). We then define the $(2i-1)$-dimensional winding number to be the pull-back of the generator
$$
W_{2i-1}(U) := U^*W_{2i-1} \in H^{2i-1}(X, Y; \Z).
$$

\begin{lem} \label{lem:first_obstruction}
Let $X$ be a finite CW complex, and $Y \subset X$ a subcomplex. A continuous map $u \in C((X, Y), (U(1), 1))$ is (relatively) homotopic to the constant map at $1$, if and only if $W_1(u) = 0$.
\end{lem}

\begin{proof}
The ``if'' part is clear. For the ``only if'' part (cf. \cite{DeNittisGomi18}), a standard obstruction theory argument can be applied: Because of the assumptions about $X$ and $Y$, $u \in C((X, Y), (U(1), 1))$ is relatively homotopic to the constant map at $1$, if and only if so is the the map $\bar{u} \in C((X/Y, Y/Y), (U(1), 1))$ induced from $u$, where $X/Y$ is the CW complex given by collapsing $Y$ to a point. Accordingly, we can assume $Y = \pt$ is a point (a $0$-cell) from the beginning. For $k = 0, 1, 2, \cdots$, we denote by $X_k$ the $k$-skeleton of the CW complex $X$. Thus, $X_0$ consists of all the $0$-cells, and $X_k$ is given by attaching the boundary of each $k$-cell $e^k$ to $X_{k-1}$. 

Because $U(1)$ is connected, there is a path between $1 \in U(1)$ and $u(e^0) \in U(1)$ for each $0$-cell $e^0$. For the base $0$-cell $\pt$, we choose the path to be the constant. Such paths together define a relative homotopy between $u|_{X_0} : X_0 \to U(1)$ and the constant map at $1$. By the homotopy extension property, we can extend the relative homotopy on $X_0$ to one between $u : X \to U(1)$ and a map $u_1 : X \to U(1)$ such that $u_1|_{X_0} \equiv 1$. Now, each $1$-cell $e^1$ defines a loop $u_1 : e^1/\partial e^1 \to U(1)$ based at $1 \in U(1)$. Hence its winding number defines a $1$-cocycle of the cellular cochain complex $C^1(X, \pt; \Z)$. This represents $W_1(u_1) = W_1(u) \in H^1(X, \pt; \Z)$, in view of the case that $X = U(1)$. The assumption $W_1(U) = 0$ says that $u_1 : e^1/\partial e^1 \to U(1)$ is homotopic to the constant loop at $1$. Such homotopies together define a relative homotopy between $u_1|_{X_1} : X_1 \to U(1)$ and the constant map. By the homotopy extension property, it extends to a relative homotopy between $u_1 : X \to U(1)$ and $u_2 : X \to U(1)$ such that $u_2|_{X_1} \equiv 1$. Then, each $2$-cell $e^2$ defines an element $u_2 : e^2/\partial e_2 \to U(1)$. Since $\pi_2(U(1)) = 0$, each map $e^2/\partial e_2 \to U(1)$ is homotopic to the constant map, and such homotopies together constitute a homotopy from $u_2|_{X_2} : X_2 \to U(1)$ to the constant map. By the homotopy extension property, this homotopy extends one between $u_2 : X \to U(1)$ and $u_3 : X \to U(1)$ such that $u_3|_{X_2} \equiv 1$. Because $\pi_i(U(1)) = 0$ for $i \ge 2$, we can repeat the same argument to get a homotopy from $u_{i-1} : X \to U(1)$ to $u_i : X \to U(1)$ such that $u_i|_{X_{i-1}} \equiv 1$. Because $X$ is a finite complex, this procedure terminates at a finite step, yielding a homotopy to the constant map on $X$. Putting all the homotopies together, we get a homotopy from $u$ to the constant map on $X$. 
\end{proof}

\begin{prop} \label{prop:obstruction}
Let $X$ be a finite CW complex which contains only cells of dimension $3$ or less, and $Y \subset X$ a subcomplex. Let $N \ge 2$. A continuous map $U \in C((X, Y), (U(N), 1))$ is (relatively) homotopic to the constant map at $1$, if and only if $W_1(U) = 0$ and $W_3(U) = 0$.
\end{prop}

\begin{proof}
The ``only if'' part is clear. For the ``if'' part, Lemma~\ref{lem:exact_sequence} and Lemma~\ref{lem:first_obstruction} imply that the given map $U$ is relatively homotopic to a map in $U' \in C((X, Y), (SU(N), 1))$. We have $W_3(U') = W_3(U)$. Hence it suffices to show that $W_3(U') = 0$ implies that $U'$ is relatively homotopic to the constant map. Then its proof is essentially the same as that of Lemma~\ref{lem:first_obstruction}: We can assume that $Y$ is a $0$-cell. Since $\pi_i(SU(N)) = 0$ for $i \le 2$, the map $U'$ is relatively homotopic to $U'' : X \to SU(N)$ such that $U''|_{X_2} \equiv 1$. Then, we have $\pi_3(SU(N)) \cong \Z$, and the map $U'' : e^3/\partial e^3 \to SU(N)$ defines a cellular $3$-cocycle which represents $W_3(U'') = W_3(U') = W_3(U)$. The vanishing $W_3(U) = 0$ ensures that $U''$ is homotopic to the constant map. 
\end{proof}

So far, we are in the topological setup, so that given maps and their homotopy are continuous. When the given CW complexes are smooth manifolds and given maps are smooth, then, by approximation, their (continuous) homotopy can be replaced by a smooth homotopy (through a homotopy of homotopies). In this paper, this replacement may be implicitly adapted. 

\medskip

As is mentioned in the introduction, when we are interested in the classification of (topological invariants of) quantum systems on $2$-dimensional lattices subject to a periodic driving, we would like to know the obstruction for $U \in C((T^2 \times S^1, T^2 \times \{ 0 \}), (U(N), 1))$ to being trivial. For a compact oriented $2$-dimensional manifold $\Sigma$ without boundary, it holds that
$$
\Z \cong H^1(\Sigma \times S^1, \Sigma \times \{ 0 \}; \Z) 
\subset H^1(\Sigma \times S^1; \Z),
$$
and this subgroup is generated by the pull-back of $H^1(S^1; \Z) \cong \Z$ under the projection $\Sigma \times S^1 \to S^1$. This implies that $W_1(U)$ is computed as the winding number of $\det U|_{\{ x \} \times S^1} : \{ x \} \times S^1 \to U(1)$, where $x \in \Sigma$ is any point. It also holds that
$$
\Z \cong H^3(\Sigma \times S^1, \Sigma \times \{ 0 \}; \Z) 
\cong H^3(\Sigma \times S^1; \Z).
$$
Hence the relative $3$-dimensional winding number $W_3(U)$ agrees with the absolute $3$-dimensional winding number. As a matter of fact, a compact oriented manifold (without boundary) admits a CW decomposition (see \cite{Milnor16} for example), and we can apply Proposition~\ref{prop:obstruction}. Then the map $U$ is relatively homotopic to the constant map at $1$, if and only if the $1$-dimensional winding number along a point $x \in \Sigma$ and the (absolute) $3$-dimensional winding number $W_3(U) \in H^3(\Sigma \times S^1; \Z)$ are vanishing. The $1$-dimensional winding number is easier to compute, and if it is non-trivial, then we can conclude that $U$ is non-trivial. If the $1$-dimensional winding number is trivial, then the remaining obstruction is $W_3(U)$. Thanks to the exact sequence in Lemma~\ref{lem:exact_sequence}, we can assume in this case that $U$ takes values in $SU(N)$. Instead if $W_1(U) = p \neq 0$ we consider $U_p = U\cdot \mathrm{diag}(\ee^{-2 \pi \ii p t}, 1, \ldots, 1)$ that satisfies $W_1(U_p) =0$ and $W_3(U_p) = W_3(U)$  by additivity of the winding numbers. In particular $U_p$ is homotopic to an $SU(N)$-valued map and shares the same value for $W_3$.This motivates us to give a local expression of the $3$-dimensional winding number for $SU(N)$-valued maps.

\section{Further examples \label{app:exmp}}


\subsection{The adjoint $S^2 \times S^1 \to SU(2)$}
\label{subsec:adjoint}

Let $T \subset SU(2)$ be the maximal torus consisting of diagonal matrices, which is diffeomorphic to the circle $S^1 \cong U(1)$. The quotient space $SU(2)/T$ is readily identified with the $2$-dimensional sphere $S^2 = \C P^1$ by
$$
\left(
\begin{array}{rr}
u & -\bar{v} \\
v & \bar{u}
\end{array}
\right)/T
\mapsto
[u : v].
$$
By the adjoint action, we have a smooth surjective map
\begin{align*}
U &: SU(2)/T \times T \to SU(2), &
(gT, h) &\mapsto ghg^{-1},
\end{align*}
This map gives rise to a double covering over $SU(2) \backslash \{ \pm 1 \}$, but not over the whole of $SU(2)$. Because $SU(2)/T \times T \cong S^2 \times S^1$ is a compact oriented $3$-dimensional manifold without boundary, the $3$-dimensional winding number $W_3(U)$ makes sense. This number agrees with the mapping degree of $U$. It is known \cite{Atiah65} that $W_3(U) = 2$, which we compute through our local formula.

It is easy to see the eigenvalue crossings:
$$
\begin{array}{|c|c|c|c|c|c|}
\hline
j & \mbox{eigenvalues} & 
\lambda_1 & \lambda_2 & \lambda_1 - 1& \Cr_j(U) \\
\hline
1 & 1, 1 & 0 & 0 & -1 & S^2 \times \{ 1 \} \\
\hline
2 & -1, -1 & 1/2 & -1/2 & -1/2 & S^2 \times \{ -1 \} \\
\hline
\end{array}
$$

To apply Theorem \ref{thm:main} for $j = 1$, we choose a closed tubular neighborhood $N_1$ of $\Cr_1(U) = S^2 \times \{ 1 \}$ to be
$
N_1 =
S^1 \times
\{ \ee^{2\pi \ii t} |\ -1/4 \le t \le 1/4 \}
\cong S^1 \times [-1/4, 1/4].
$
At $([u: v], \pm 1/4)) \in \partial N_1$, the value of $U$ is 
$$
U([u : v], \pm 1/4)
=
\left(
\begin{array}{rr}
u & -\bar{v} \\
v & \bar{u}
\end{array}
\right)
\left(
\begin{array}{cc}
\pm i & 0 \\
0 & \mp i
\end{array}
\right)
\left(
\begin{array}{rr}
u & -\bar{v} \\
v & \bar{u}
\end{array}
\right)^{-1},
$$
so that its eigenvalues are $\ee^{2\pi \ii \lambda_1} = i$ and $\ee^{2\pi \ii \lambda_2} = -i$: $\lambda_1 = \tfrac{1}{4} \geq \lambda_2 = -\tfrac{1}{4} \geq \lambda_1-1 = -\tfrac{3}{4}$.
Thus, on the connected component $S^2 \times \{ 1/4 \} \subset \partial N_1$, the eigenvector of $U([u: v], 1/4)$ with eigenvalue $\ee^{2\pi \ii \lambda_1} = i$ is
$$
\left(
\begin{array}{rr}
u & -\bar{v} \\
v & \bar{u}
\end{array}
\right)
\left(
\begin{array}{c}
1 \\ 0
\end{array}
\right)
=
\left(
\begin{array}{c}
u \\ v
\end{array}
\right),
$$
which spans the tautological line bundle on $S^2 = \C P^1$. On the other connected component $S^2 \times \{ -1/4 \} \subset \partial N_1$, the eigenvector of $U([u: v], -1/4)$ with eigenvalue $\ee^{2\pi \ii \lambda_1} = i$ is
$$
\left(
\begin{array}{rr}
u & -\bar{v} \\
v & \bar{u}
\end{array}
\right)
\left(
\begin{array}{c}
0 \\ 1
\end{array}
\right)
=
\left(
\begin{array}{c}
-\bar{v} \\ \bar{u}
\end{array}
\right),
$$
which spans the dual of the tautological line bundle on $S^2 = \C P^1$. Taking the induced orientations on $S^2 \times \{ \pm 1/4 \}$ into account, we find that $\Ch(\Cr_1; 1) = -1 - (+1) = -2$. Hence Theorem \ref{thm:main} gives $W_3(U) = 2$, as anticipated. The application of Theorem \ref{thm:main} for $j = 2$ is similar.


\subsection{The standard embedding $SU(2) \to SU(3)$}

Let $U : SU(2) \to SU(3)$ be the standard embedding
$$
U(
\left(
\begin{array}{rr}
u & -\bar{v} \\
v & \bar{u}
\end{array}
\right)
)
=
\left(
\begin{array}{rrr}
u & -\bar{v} & 0 \\
v & \bar{u} & 0 \\
0 & 0 & 1
\end{array}
\right).
$$
The winding number is $W_3(U) = 1$, as can be computed directly. We here compute this number by means of the results in this note. Using the unique expression of the eigenvalues $\ee^{2\pi \ii \lambda_1}, \ee^{2\pi \ii \lambda_2}, \ee^{2\pi \ii \lambda_3}$ of a matrix in $SU(3)$ in terms of $\lambda_i \in \R$ such that $\lambda_1 + \lambda_2 + \lambda_3 = 0$ and $\lambda_1 \ge \lambda_2 \ge \lambda_3 \ge \lambda_1 - 1$, we can summarize the crossings of the eigenvalues as follows:
$$
\begin{array}{|c|c|c|c|c|c|c|}
\hline
\mbox{eigenvalues} & 
\lambda_1 & \lambda_2 & \lambda_3 & \lambda_1 - 1& \mbox{subspace in $SU(2)$} \\
\hline
1, 1, 1 & 0 & 0 & 0 & -1 & U^{-1}(\1_{3}) = \{ \1_{2} \} \\
\hline
-1, 1, -1 & 1/2 & 0 & -1/2 & -1/2 & \Cr_3(U) = \{ -\1_{2} \} \\
\hline
\end{array}
$$
Hence we can apply Theorem \ref{thm:full_degeneracy_at_1}:
\begin{itemize}
	\item
	For $j = 1$, we have $\Cr_1(U) = \emptyset$ and $W_3(U) =  \Ch(\{ \1_{2} \}; 2) + \Ch(\{ \1_{2} \}; 3)$.
	
	\item
	For $j = 2$, we have $\Cr_2(U) = \emptyset$ and $W_3(U) = \Ch(\{ \1_{2} \}; 3)$.
	
	\item
	For $j = 3$, we have $\Cr_3(U) = \{ -\1_{2} \}$ and $W_3(U) = -\Ch(\{-\1_{2}\}; 3)$.
\end{itemize}
To compute the local indices, we can use the $3$-dimensional disks $D_x$ containing $x = 1$ and $D_{x'}$ containing $x' = -1$ in $SU(2)$. Thus, all the relevant indices are Chern number of some line bundles over $\partial D_x = \partial D_{x'}$. The eigenvalues of $U \in \partial D_x = \partial D_{x'}$ are $i$, $1$ and $-i$. Note that
$$
\overbrace{\lambda_1}^{1/4} >
\overbrace{\lambda_2}^{0} >
\overbrace{\lambda_3}^{-1/4} >
\overbrace{\lambda_1 - 1}^{-3/4}.
$$
For $j = 1$, the local index is the Chern number of the tensor product of the line bundles whose fibers are eigenspaces with eigenvalues $\ee^{2\pi \ii \lambda_2} = 1$ and $\ee^{2\pi \ii \lambda_3} = -i$. The Chern number of the latter line bundle is computed in Example~\ref{exmp:identity}, whereas that of the former is trivial, since, for any $y \in \partial D_x = \partial D_{x'}$, the eigenvector of $U(y) \in SU(3)$ with eigenvalue $1$ is $(0,\, 0,\, 1)^t$.
Therefore we get $\Ch(\{ \1_{2} \}; 2)=0$ and $\Ch(\{ \1_{2} \}; 3) = 1$. Finally, for $j = 3$, the local index is also the Chern number of the line bundle whose fibers are eigenspaces with eigenvalues $-i$, so that $\Ch(\{-\1_{2}\}; 3) = -1$ by the computation in Example~\ref{exmp:identity}.


\subsection{A perturbed embedding $SU(2) \to SU(3)$}

Let us consider a family of embedding
\begin{align*}
U &: SU(2) \to SU(3), &
U(
\left(
\begin{array}{rr}
u & -\bar{v} \\
v & \bar{u}
\end{array}
\right)
)
&=
\left(
\begin{array}{ccc}
u \ee^{\ii t} & -\bar{v}\ee^{\ii t} & 0 \\
v \ee^{\ii t} & \bar{u}\ee^{\ii t} & 0 \\
0 & 0 & \ee^{-2\ii t}
\end{array}
\right)
\end{align*}
parametrized by $t \in \R$. The $3$-dimensional winding number of $U$ is $W_3(U) = 1$ for any $t$, since $U$ at $t$ is homotopic to $U$ at $t = 0$, which is the standard embedding. The three eigenvalues of $U$ are generally expressed as $\{ \ee^{\ii t}u, \ee^{\ii t}\bar{u}, \ee^{-2\ii t} \}$, where $u \in U(1)$ and $t \in \R$.

As a special choice, we take $\ee^{\ii t} = i$, so that
$$
U(
\left(
\begin{array}{rr}
u & -\bar{v} \\
v & \bar{u}
\end{array}
\right)
)
=
\left(
\begin{array}{ccc}
iu & -i\bar{v} & 0 \\
iv  & i\bar{u} & 0 \\
0 & 0 & -1
\end{array}
\right).
$$
In this case, the three eigenvalues are distinct, or two of them coincide. The following table summarizes the detail of the latter case by using the unique expression of the eigenvalues $\ee^{2\pi \ii \lambda_1}, \ee^{2\pi \ii \lambda_2}, \ee^{2\pi \ii \lambda_3}$ of a matrix in $SU(3)$ in terms of $\lambda_i \in \R$ such that $\lambda_1 + \lambda_2 + \lambda_3 = 0$ and $\lambda_1 \ge \lambda_2 \ge \lambda_3 \ge \lambda_1 - 1$.
$$
\begin{array}{|c|c|c|c|c|c|c|}
\hline
j & \mbox{eigenvalues} & 
\lambda_1 & \lambda_2 & \lambda_3 & \lambda_1 - 1& \Cr_j(U) \\
\hline
1 & i, i, 1 & 1/4 & 1/4 & -1/2 & -3/4 & \pt \\
\hline
2 & -1, -i, -i & 1/2 & -1/4 & -1/4 & -1/2 & \pt \\
\hline
3 & -1, 1, -1 & 1/2 & 0 & -1/2 & -1/2 & S^2 \\
\hline
\end{array}
$$
Note that $\Cr_1(U) = \{ \1 \}$, $\Cr_2(U) = \{ -\1 \}$ and
$$
\Cr_3(U) =
\left\{
\left(
\begin{array}{rr}
u & -\bar{v} \\
v & \bar{u}
\end{array}
\right) \in SU(2)
\bigg|
u + \bar{u} = 0
\right\}.
$$
Accordingly, we can apply Theorem \ref{thm:main}: The application of the theorem for $j = 1, 2$ reduces to the calculations of the identity map $SU(2) \to SU(2)$ given in Example~\ref{exmp:identity}, so we omit the detail. To apply Theorem \ref{thm:main} for $j = 3$, we choose a closed tubular neighborhood $N$ of the $2$-dimensional sphere $\Cr_3(U) \subset SU(2)$ to be
$$
N = 
\left\{
\left(
\begin{array}{rr}
u & -\bar{v} \\
v & \bar{u}
\end{array}
\right) \in SU(2)
\bigg|
- \sqrt{2} \le u + \bar{u} \le \sqrt{2}
\right\}.
$$
We can identify $S^2 \times [-1/\sqrt{2}, 1/\sqrt{2}]$ with $N$ by
$$
(X, Y, Z, t) \mapsto
\left(
\begin{array}{rr}
t + i\sqrt{1-t^2}X & -\sqrt{1-t^2}(Y - iZ) \\
\sqrt{1-t^2}(Y + iZ) & t - i \sqrt{1 - t^2}X
\end{array}
\right),
$$
where $S^2 = \{ (X, Y, Z) \in \R^3 |\ X^2 + Y^2 + Z^2 = 1 \}$. Thus, on the boundary $\partial N = S^2 \times \{ \pm 1/\sqrt{2} \}$, we have
$$
U(X, Y, Z, \pm 1/\sqrt{2})
=
\left(
\begin{array}{ccc}
\frac{-X \pm i}{\sqrt{2}} & \frac{-Z - iY}{\sqrt{2}} & 0 \\
\frac{-Z + iY}{\sqrt{2}} & \frac{X \pm i}{\sqrt{2}} & 0 \\
0 & 0 & -1
\end{array}
\right).
$$
The eigenvalues of this matrix are as follows:
$$
\begin{array}{|c|c|c|c|c|c|}
\hline
t & \mbox{eigenvalues} & 
\lambda_1 & \lambda_2 & \lambda_3 & \lambda_1 - 1 \\
\hline
t = \frac{1}{\sqrt{2}} &
\frac{-1+i}{\sqrt{2}}, \frac{1 + i}{\sqrt{2}}, -1 & 
3/8 & 1/8 & -1/2 & -5/8 \\
\hline
t = -\frac{1}{\sqrt{2}} &
-1, \frac{1-i}{\sqrt{2}}, \frac{-1 - i}{\sqrt{2}} & 
1/2 & -1/8 & -3/8 & -1/2 \\
\hline
\end{array}
$$
The local index $\Ch(\Cr_3; 3)$ is the Chern number of the line bundle whose fibers are eigenspaces with eigenvalues $\ee^{2\pi \ii \lambda_3}$. Thus, on $S^2 \times \{ 1/\sqrt{2} \}$, we have the constant eigenvector with eigenvalue $-1$, so that the Chern number of the line bundle is trivial. On $S^2 \times \{ -1/\sqrt{2} \}$, the line bundle is non-trivial. At $g = (X, Y, Z, -1/\sqrt{2})$ with $X \neq 1$, an eigenvector $v_3^-(g)$ of $U(g)$ with eigenvalue $\ee^{2\pi \ii \lambda_3} = (-1 - i)/\sqrt{2}$ is given by
$$
v_3^-(g) =
\left(
\begin{array}{c}
\frac{Z + iY}{1 - X} \\ 1 \\ 0
\end{array}
\right).
$$
At $g = (X, Y, Z, -1/\sqrt{2})$ with $X \neq -1$, an eigenvector $v_3^+(g)$ of $U(g)$ with eigenvalue $\ee^{2\pi \ii \lambda_3} = (-1 - i)/\sqrt{2}$ is given by
$$
v_3^+(g) =
\left(
\begin{array}{c}
1 \\ \frac{Z - iY}{1 + X} \\ 0
\end{array}
\right).
$$
On the circle in the sphere $S^2 \times \{ -1/\sqrt{2} \}$
$$
\{ g = (X, Y, Z, -1/\sqrt{2}) |\ X = 0, Y^2 + Z^2 = 1 \}
\subset S^2 \times \{ -1/\sqrt{2} \},
$$
we have a $U(1)$-valued map $f(g) = Z - iY = -i(Y + iZ)$ which measures the discrepancy of $v_3^+(g)$ and $v_3^-(g)$ by $v_3^+(g) = f(g)v_3^-(g)$. This implies that the Chern number of the line bundle over $S^2 \times \{ -1/\sqrt{2} \}$ is $1$ under a choice of an orientation. To summarize, we get $\Ch(\Cr_3; 3) = 0-1 = -1$ and $W_3(U) = 1$.


\subsection{The adjoint and embedding $S^2 \times S^1 \to SU(3)$}
\label{subsec:adjoint_and_embedding}

We here consider the map $U : S^2 \times S^1 \to SU(3)$ 
$$
U([u : v], \ee^{\ii t})
=
\left(
\begin{array}{rrr}
u & -\bar{v} & 0 \\
v & \bar{u} & 0 \\
0 & 0 & 1
\end{array}
\right)
\left(
\begin{array}{ccc}
\ee^{\ii t} & 0 & 0 \\
0 & \ee^{-\ii t} & 0 \\
0 & 0 & 1
\end{array}
\right)
\left(
\begin{array}{rrr}
u & -\bar{v} & 0 \\
v & \bar{u} & 0 \\
0 & 0 & 1
\end{array}
\right)^{-1}
$$
given by composing the adjoint map $S^2 \times S^1 \to SU(2)$ in \S\S\ref{subsec:adjoint} and the standard embedding $SU(2) \to SU(3)$. Although $W_3(U) = 2$ is clear, we consider to apply Theorem \ref{thm:full_degeneracy_at_1}.

The eigenvalue crossings are as follows:
$$
\begin{array}{|c|c|c|c|c|c|}
\hline
\mbox{eigenvalues} & \lambda_1 & \lambda_2 & \lambda_3 & \lambda_1 - 1 &
\mbox{spaces in $S^2 \times S^1$} \\
\hline
1, 1, 1 & 0 & 0 & 0 & -1 & U^{-1}(\1) = S^2 \times \{ 1 \} \\
\hline
-1, 1, -1 & 1/2 & 0 & -1/2 & -1/2 & \Cr_3(U) = S^2 \times \{ -1 \} \\
\hline
\end{array}
$$
Note that $\Cr_1(U) = \emptyset$ and $\Cr_2(U) = \emptyset$. Theorem \ref{thm:full_degeneracy_at_1} produces:
\begin{itemize}
	\item
	for $j = 1$, we have $W_3(U) = \Ch(U^{-1}(\1); 2)+\Ch(U^{-1}(\1); 3)$,
	
	\item
	for $j = 2$, we have $W_3(U) = \Ch(U^{-1}(\1); 3)$, 
	
	\item
	for $j = 3$, we have $W_3(U) = - \Ch(\Cr_3; 3)$.

\end{itemize}
It turns out that the all the calculations of the local indices reduce to those given in Example~\ref{subsec:adjoint}. Hence we just consider the case of $j = 1$. In this case, we choose $N_1 = S^2 \times [-1/4, 1/4]$ in Example~\ref{subsec:adjoint} as the closed tubular neighborhood of $U^{-1}(\1) = S^2 \times \{ 1 \} \subset S^2 \times S^1$. On the boundary $\partial N_1 = S^2 \times \{ \pm 1/4 \}$, the map $U$ takes the values
$$
U([u : v], \ee^{\ii t})
=
\left(
\begin{array}{rrr}
u & -\bar{v} & 0 \\
v & \bar{u} & 0 \\
0 & 0 & 1
\end{array}
\right)
\left(
\begin{array}{ccc}
\pm i& 0 & 0 \\
0 & \mp i& 0 \\
0 & 0 & 1
\end{array}
\right)
\left(
\begin{array}{rrr}
u & -\bar{v} & 0 \\
v & \bar{u} & 0 \\
0 & 0 & 1
\end{array}
\right)^{-1},
$$
hence its eigenvalues are $\ee^{2\pi \ii \lambda_1} = i$, $\ee^{2\pi \ii \lambda_2} = 1$ and $\ee^{2\pi \ii \lambda_3} = -i$
$$
\overbrace{\lambda_1}^{1/4} \ge
\overbrace{\lambda_2}^{0} \ge
\overbrace{\lambda_3}^{-1/4} \ge
\overbrace{\lambda_1-1}^{-3/4}.
$$
On the connected component $S^2 \times \{ 1/4 \} \subset \partial N_1$, we can find the following eigenvectors with eigenvalues $\ee^{2\pi \ii \lambda_2} = 1$ and $\ee^{2\pi \ii \lambda_3} = -i$, respectively $(0,\,0,\,1)^t$ and $(-\bar v,\,\bar u ,\,1)^t$.
Hence the tensor product of the line bundles whose fibers are the eigenspaces with eigenvalues $\ee^{2\pi \ii \lambda_2} = 1$ and $\ee^{2\pi \ii \lambda_3} = -i$ is the dual of the tautological line bundle on $S^2 = \C P^1$. On the other connected component $S^2 \times \{ -1/4 \} \subset \partial N_1$, we have the following eigenvectors with eigenvalues $\ee^{2\pi \ii \lambda_2} = 1$ and $\ee^{2\pi \ii \lambda_3} = -i$, respectively $(0,\,0,\,1)^t$ and $(u,\,v,\,1)^t$.
Then the tensor product of the line bundles whose fibers are the eigenspaces with eigenvalues $\ee^{2\pi \ii \lambda_2} = 1$ and $\ee^{2\pi \ii \lambda_3} = -i$ is the tautological line bundle on $S^2 = \C P^1$. Taking the orientation into account, we find $\Ch(U^{-1}(\1); 2)+\Ch(U^{-1}(\1); 3) = 2$.

\subsection{Floquet map example}
\label{subsec:example_top}

Let $U_i : \C P^1 \times [0, 1] \to SU(2)$ be the following maps
\begin{align*}
U_1([u : v], t)
&=
\left(
\begin{array}{rr}
u & -\bar{v} \\
\bar{v} & u
\end{array}
\right)
\left(
\begin{array}{cc}
\ee^{\pi \ii t/2} & 0 \\
0 & \ee^{-\pi \ii t/2}
\end{array}
\right)
\left(
\begin{array}{rr}
u & -\bar{v} \\
\bar{v} & u
\end{array}
\right)^{-1}, 
\\
U_2([u : v], t)
&=
\left(
\begin{array}{rr}
u & -\bar{v} \\
\bar{v} & u
\end{array}
\right)
\left(
\begin{array}{cc}
\ee^{-3\pi \ii t/2} & 0 \\
0 & \ee^{3\pi \ii t/2}
\end{array}
\right)
\left(
\begin{array}{rr}
u & -\bar{v} \\
\bar{v} & u
\end{array}
\right)^{-1}.
\end{align*}
Note that, as $t \in [0, 1]$ varies from $t = 0$ to $t = 1$, the element $\ee^{\pi \ii t/2}$ travels on $U(1)$ aticlockwisely from $1 \in U(1)$ to $i \in U(1)$, whereas $\ee^{-3\pi \ii t} = \ee^{2\pi \ii (- 3/4t)}$ travels clockwisely from $1 \in U(1)$ to $i \in U(1)$ via $-1 \in U(1)$. They satisfy $U_i|_{\C P^1 \times \{ 0 \}} = \1$ and 	$U_i(\C P^1 \times \{ 1 \}) \subset SU(2)_{\le 0}$.
%
%
%
Further, $U_i(\C P^1 \times (0, 1]) \subset SU(2)_{\le 1}$ and $U_1|_{\C P^1 \times \{ 1 \}} = U_2|_{\C P^2 \times \{ 1 \}}$. It is easy to see $\Cr_1(U) = \C P^1 \times \{ 0 \}$ and $\Cr_2(U) = \emptyset$ for $U_1$. We can compute the topological numbers of $U_1$ as follows:
\begin{align*}
\Top(U_1; 1) &= -\Ch(\Cr_1; 1) = 1, &
\Top(U_1; 2) &= -\Ch(\Cr_2; 2) = 0.
\end{align*}
For $U_2$, we have $\Cr_1(U) = \C P^1 \times \{ 0 \}$ and $\Cr_2(U) = \C P \times \{ 2/3 \}$. We can also compute
\begin{align*}
\Top(U_2; 1) &= - \Ch(\Cr_1; 1) = -1, &
\Top(U_2; 2) &= - \Ch(\Cr_2; 2) = -2.
\end{align*}
These computations show that $U_1$ and $U_2$ are not homotopic relative to $\C P^1 \times \partial [0, 1]$. Notice that $U_1 \cup U_2 : \C P^1 \times S^1 \to SU(2)$ is just the map in Example~\ref{subsec:adjoint}. Hence $\Top(U_1; j) - \Top(U_2; j) = W_3(U_1 \cup U_2) = 2$ as anticipated. 

Composing the two maps $U_i : \C P^1 \times [0, 1] \to SU(2)$ and the standard embedding $SU(2) \to SU(3)$, we define $V_i : \C P^1 \times [0, 1] \to SU(3)$. In this case, $V_i^{-1}(1) = \C P^1 \times \{ 0 \}$. We have $\Cr_1(U) = \emptyset$ for $V_1$ and $V_2$. Thus, to $\Top(V_i; 1)$, the contributions $\Ch(\Cr_1; 1) = 0$ from the simple crossing $\Cr_1(U)$ is trivial for both $V_1$ and $V_2$. But, the contributions $\Ch(V_i^{-1}(\1);2)$ and $\Ch(V_i^{-1}(\1);3)$ from the full crossings at $\C P^1 \times \{ 0 \}$ are non-trivial, and we get
\begin{align*}
\Top(V_1; 1) &= - \Ch(\Cr_1; 1) + \Ch(V_1^{-1}(\1);2)+\Ch(V_1^{-1}(\1);3) = 1, \\
\Top(V_2; 1) &= - \Ch(\Cr_2; 1) +\Ch(V_1^{-1}(\1);2)+ \Ch(V_1^{-1}(\1);3) = -1.
\end{align*}
This contribution from the full degeneracy at $t = 0$ seems not considered in \cite{NathanRudner15}, but the above example shows that it plays an indispensable role in the topological invariant.

\bigskip \bigskip

{\footnotesize

\begin{tabular}{ll}
(K. Gomi) & \textsc{Department of Mathematical Sciences, Shinshu University} \\
 &   Matsumoto, Nagano 390-8621, Japan \\
 &  {E-mail address}: \href{mailto:kgomi@math.shinshu-u.ac.jp}{\texttt{kgomi@math.shinshu-u.ac.jp}} \\
\\
(C. Tauber) & \textsc{Institute for Theoretical Physics, ETH Z\"{u}rich} \\
 & Wolfgang-Pauli-Str. 27, CH-8093 Z\"{u}rich, Switzerland \\
 &  {E-mail address}: \href{mailto:tauberc@phys.ethz.ch}{\texttt{tauberc@phys.ethz.ch}} \\
\end{tabular}

}
 
\end{document}